\documentclass[a4paper,USenglish,cleveref, autoref, thm-restate]{lipics-v2021}
\usepackage[ruled]{algorithm2e}
\usepackage{xspace}
\usepackage{bm}

\hideLIPIcs  

\graphicspath{{./graphics/}}

\title{Smallest Enclosing Disk Queries Using Farthest-Point Voronoi Diagrams}

\author{Kevin Buchin}{Department of Computer Science, TU Dortmund University, Germany \and \url{https://ae.cs.tu-dortmund.de/team/buchin/}}{}{https://orcid.org/0000-0002-3022-7877}{}%

\author{Mark Joachim Krallmann}{Department of Computer Science, TU Dortmund University, Germany}{}{https://orcid.org/0009-0003-8927-4082}{Funded by the Deutsche Forschungsgemeinschaft (DFG, German Research Foundation) - 550797858}

\author{Frank Staals}{Department of Information and Computing Sciences, Utrecht University, The Netherlands}{}{https://orcid.org/0009-0004-8522-1351}{}

\authorrunning{K.\,Buchin, M.\,J.\,Krallmann, and F. Staals} 

\Copyright{Kevin Buchin, Mark J. Krallmann, and Frank Staals} 

\ccsdesc[500]{Theory of computation~Computational geometry}
\ccsdesc[300]{Theory of computation~Design and analysis of algorithms}
\ccsdesc[300]{Theory of computation~Data structures design and analysis}

\keywords{Range searching, smallest enclosing disk, farthest point
  Voronoi diagram} 

\category{} 

\relatedversiondetails[cite=esaVersion]{ESA 2026 Conference Paper}{https://doi.org/10.4230/LIPIcs.ESA.2026.124}

\nolinenumbers 

\EventEditors{Philip Bille, Seth Pettie, and Sabine Storandt}
\EventNoEds{3}
\EventLongTitle{34th Annual European Symposium on Algorithms (ESA 2026)}
\EventShortTitle{ESA 2026}
\EventAcronym{ESA}
\EventYear{2026}
\EventDate{August 31--September 4, 2026}
\EventLocation{L'Aquila, Italy}
\EventLogo{}
\SeriesVolume{388}
\ArticleNo{122}
\DeclareMathOperator{\fpvdmth}{FPVD}
\DeclareMathOperator{\cellmth}{FPVC}
\DeclareMathOperator{\ch}{CH}
\DeclareMathOperator{\sedmth}{SED}
\DeclareMathOperator{\midmth}{mid}
\DeclareMathOperator{\fcases}{FC}
\DeclareMathOperator{\smd}{SMD}
\DeclareMathOperator{\md}{MD}
\DeclareMathOperator{\nextmth}{next}
\DeclareMathOperator{\prevmth}{prev}
\DeclareMathOperator*{\polylog}{polylog}
\newtheorem{assumption}[theorem]{Assumption}
\newtheorem{fact}[theorem]{Fact}

\newcommand{\ie}{i.\@e.\ }
\newcommand{\eg}{e.\@g.\ }
\newcommand{\Eg}{E.\@g.\ }
\newcommand {\mathset} [1] {\ensuremath {\mathbb {#1}}\xspace}
\newcommand {\R} {\mathset {R}}

\newcommand{\myremark}[4]{\textcolor{blue}{\textsc{#1 #2: }}\textcolor{#4}{\textsf{#3}}}
\renewcommand{\myremark}[4]{}

\begin{document}
\maketitle
\begin{abstract}
Let $S$ be a set of $n$ points in $\R^2$. 
Our goal is to preprocess $S$ to efficiently compute the smallest enclosing disk of the points in $S$ that lie inside an axis-aligned query rectangle. 
Previous data structures for this problem achieve a query time of $O(\log^6 n)$ with $O(n \log^2 n)$ preprocessing time and space by lifting the points to 3D, dualizing them into polyhedra, and searching through their intersections.
We present a significantly simpler approach, solely based on 2D geometric structures, specifically 2D farthest-point Voronoi diagrams.
Our approach achieves a deterministic query time of $O(\log^4 n)$ and, via randomization, an expected query time of $O(\log^{5/2} n \log\log n)$ with the same preprocessing bounds.
\end{abstract}
\section{Introduction}
Let $S$ be a set of $n$ points in $\R^2$. Computing the smallest
radius disk that contains all points in $S$, i.e. the \emph{smallest
  enclosing disk} $\sedmth(S)$ is a fundamental and well-studied
problem in computational geometry that has various applications in
facility location and
clustering~\cite{AgarwalP02,BadoiuHI02,trsc.6.4.379}. Megiddo gave a
linear-time algorithm to compute the smallest enclosing disk (or even
the smallest enclosing ball of points in
$\R^d$)~\cite{megiddoLinearTimeAlgorithmsLinear1983}. Furthermore,
Welzl's simple, randomized incremental construction
algorithm~\cite{welzlSmallestEnclosingDisks1991} now appears in
various text books on computational
geometry~\cite{debergComputationalGeometryAlgorithms2008}. However, in
many applications we are not interested in the smallest enclosing disk
of \emph{all} points in $S$, but only of the points in an area of
interest. For example, a data analyst may want to interactively select
some region $Q$, and quickly compute the smallest enclosing disk of
the points in $S \cap Q$. See \autoref{fig:example}. Hence, we want to
build a data structure that can answer such \emph{smallest enclosing
  disk queries} efficiently.

\begin{figure}
    \centering
    \includegraphics*[]{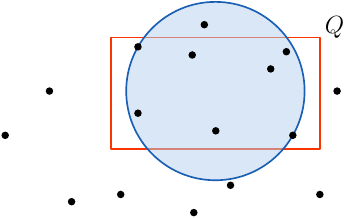}
    \caption{A rectangular query region $Q$ and the smallest enclosing disk of $S \cap Q$.}
    \label{fig:example}
\end{figure}

Our smallest enclosing disk queries are a form of \emph{range
  aggregation queries}: i.e. store the set $S$ so that one can
efficiently compute some aggregation function $f(S \cap Q)$ of the
points in a query range $Q$. Simple aggregation functions such as
counting or reporting the number of points in the query range are
extremely well studied~\cite{chazelle1990ArithmeticLowerBound,debergComputationalGeometryAlgorithms2008,nekrich2009data}. More recently, there is an interest
in more advanced aggregation functions. Examples include queries that
ask for (approximate) heavy hitters and
quantiles~\cite{afshani23range_summar_queries}, color frequency
reporting~\cite{glazenburg26stric_output_sensit_color_frequen_repor},
(approximate) smallest enclosing
disks~\cite{brassRangeAggregateQueriesGeometric2013,khareImprovedBoundsSmallest2014}
(reviewed in more detail below) and
balls~\cite{huangInRangeFarthestPoint2022}, or even queries that ask
for an entire (approximate) $k$-center, $k$-median, or $k$-means
clustering of
$S \cap
Q$~\cite{abrahamsenRangeClusteringQueries2017,ohApproximateRangeQueries2018a}. Many
of these approximate results rely on composable
coresets~\cite{har-peledCoresetsKmeansKmedian2004}; i.e. the ability
to compactly summarize a subset of the points, so that we can combine
the summaries to obtain a summary of
$S \cap Q$~\cite{nekrichApproximatingRangeaggregateQueries2010}.

\subparagraph{Smallest enclosing disk queries.} We will focus on
smallest enclosing disk queries in the case that the query ranges are
axis-aligned rectangles. There are $\Theta(n^4)$ combinatorially
distinct query ranges, so clearly we could precompute and store all
answers (taking $O(n^5)$ time $O(n^4)$ space) so that we can answer
queries in optimal $O(\log n)$ time. However, this is prohibitively
large. Alternatively, in (near) linear time and space one can build a
data structure to answer range reporting queries (e.g. a range tree or
kd-tree) to report the $k$ points in $S \cap Q$ and then compute
$\sedmth(S\cap Q)$ in additional $O(k)$ time. However, as the number
of points $k$ in the query range may be large, this query time
(e.g. $O(\log n + k)$) may not be sufficient for interactive
applications.

Unfortunately, computing the smallest enclosing disk is not
decomposable; i.e. one cannot easily compute $\sedmth(A \cup B)$ from
$\sedmth(A)$ and $\sedmth(B)$. Hence, one cannot directly use similar
ideas as for answering e.g. range counting queries. Instead, in
\cite{brassRangeAggregateQueriesGeometric2013} the authors argue that
point sets $A$ and $B$ can be lifted and dualized into convex
polyhedra $P_A$ and $P_B$, so that one can compute $\sedmth(A \cup B)$
by searching in (an implicit representation of) the intersection of
$P_A$ and $P_B$. Such a searching query is somewhat similar to
answering linear programming query, and thus, given $m$ convex
polyhedra of total complexity $n$, one can answer a query in
$O(m^3\log^3 n)$
time~\cite{eppsteinDynamicThreeDimensionalLinear1992}. By combining
this machinery with a standard 2D-range tree, one can obtain
$S \cap Q$ as $m=O(\log^2 n)$ appropriately preprocessed canonical
subsets, and thus answer a smallest enclosing disk query in
$O(\log^9 n)$ time. The data structure uses $O(n\log^2 n)$ space and
can be built in $O(n\log^2 n)$
time~\cite{brassRangeAggregateQueriesGeometric2013}. In
\cite{khareImprovedBoundsSmallest2014} the authors presented a
modified range tree that allows to reduce the number of considered
canonical sets to $O(\log n)$, thereby improving the query time to
$O(\log^6 n)$.

\subparagraph{Our Results.} We will show that we can instead decompose
the problem by using farthest-point Voronoi diagrams. The farthest
point Voronoi diagram $\fpvdmth(S)$ of $S$ is a subdivision of the
plane into interiorly disjoint maximal regions so that any point $q$
in a region $\cellmth(p, S)$ has the same site $p$ that is farthest
from $q$ among $S$. We show that given the farthest-point Voronoi diagrams
 of $A$ and $B$ (in an appropriately
preprocessed form), we can find the smallest enclosing disk
$\sedmth(A \cup B)$ in $O(\log^2 n)$ time (where $n=|A|+|B|$). Our
procedure, described in
\autoref{sec:decomp-search} extends to $m$
diagrams; and runs in $O(m^2 \log^2 n)$ time. This allows us to bypass
the whole lifting and dualizing to step, as well as the somewhat
intricate searching among the intersection of 3D convex polyhedra
steps used by the earlier
approaches~\cite{brassRangeAggregateQueriesGeometric2013,
  khareImprovedBoundsSmallest2014}. Hence, this significantly
simplifies the query procedure. Furthermore, it improves the query
time to $O(\log^4 n)$. The preprocessing step also remains relatively
simple. One may furthermore wonder whether one can use randomization
to further simplify the algorithm (as was the case in the ``base''
problem). Unfortunately, we show that one cannot straightforwardly
generalize Welzl's randomized incremental construction algorithm for
computing the smallest enclosing disk of points to computing the
smallest enclosing disk of disjoint convex polygons (corresponding to
our canonical subsets). See
\autoref{sec:Randomized_Approaches}. However, we can combine our
earliest farthest-point Voronoi diagram based machinery with a
randomized dynamic programming based framework from
Eppstein~\cite{eppsteinDynamicThreeDimensionalLinear1992} to further
speed up smallest enclosing disk queries. In particular, our data
structure answers such queries in expected $O(\log^{5/2} n\log\log n)$
time. The preprocessing time and space remain $O(n\log^2
n)$.

\section{Preliminaries}
\label{sec:Preliminaries}

In this section, we briefly review some of the basic tools and
techniques that we build upon throughout the paper.
Let $S$ be a set of $n$ points in $\R^2$.
We assume that the points in $S$ are in general position,
\ie no three points are collinear and no four points are cocircular.
If this assumption is violated, we use symbolic perturbation \cite{edelsbrunnerSimulationSimplicityTechnique1990}.
Let $\ch(S) \subseteq S$ denote the vertices of the convex hull of $S$ and $\midmth(p, q)$ denote the midpoint of two points $p, q \in S$.

\subsection{Range-aggregate queries for smallest enclosing disk}\label{subsec:sed-range-aggr}

Our goal is to preprocess and store $S$, such that we can efficiently
compute the smallest enclosing disk $\sedmth(S \cap Q)$ of an axis
aligned query rectangle
$Q = [x, x'] \times [y, y'] \subseteq \mathbb R^2$. Our goal is to
answer such queries in $O(\polylog n)$ time using $O(n\polylog n)$
space.

\begin{figure}
    \centering
    \includegraphics*[page=2]{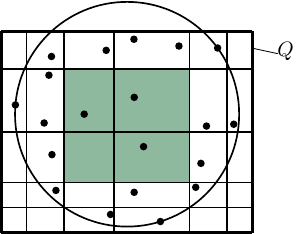}
    \caption{Based on~\cite{khareImprovedBoundsSmallest2014}. To compute $\sedmth(S \cap Q)$ we can discard the green-shaded nodes.
    }\label{fig:khare}
\end{figure}

We achieve this by building the modified 2-dimensional range tree of~\cite{khareImprovedBoundsSmallest2014} on $S$.
Compared to a standard range tree, this structure yields only $O(\log n)$ canonical nodes $v_1,  \dots, v_m$ to consider, instead of $O(\log^2 n)$.
The key observation, as described in~\cite{khareImprovedBoundsSmallest2014}, is that $\sedmth(S \cap Q)$ is determined solely by vertices of the convex hull.
All nodes that cannot contribute any point to the convex hull can be discarded, as shown in \autoref{fig:khare}.
We represent $S \cap Q$ with the corresponding canonical sets $S_1 = P(v_1), \dots, S_m = P(v_m)$ such that $\ch(S_1 \cup \dots \cup S_m) = \ch(S \cap Q)$.

We can partition $\ch(S \cap Q)$ into maximal sections contributed by one canonical set respectively.
Let $p_1, \dots, p_k$ be the points in $\ch(S)$ in clockwise order (beginning with the point that is smallest lexicographically), then we refer to the set of points $\{p_a, p_{a+1}, \dots, p_{b-1}, p_b\}$, wrapping around if needed, with $\ch(S)_{a:b}$.
We refer to $\ch(S)_{a:b}$ as a \emph{section} of $\ch(S)$.
We may include or exclude points relative to $p_a$, \eg by writing $\ch(S)_{a+1:b-1}$ we exclude $p_a$ and $p_b$.
Note that a section such as $\ch(S)_{a+1:b-1}$ may be empty if there are no points between $p_a$ and $p_b$ in the clockwise sequence.
Let $S_1, \dots, S_m$ be canonical sets.
We refer to a non-empty section $\ch(S_i)_{a:b}$ as a \emph{canonical section}
of $S_i$ when there is an equivalent section $\ch(S \cap Q)_{a':b'} = \ch(S_i)_{a:b}$ that is maximal, \ie $S_i$ does not include the clockwise and counter-clockwise neighbors, hence $\ch(S \cap Q)_{a'-1:b'} \nsubseteq S_i$ and $\ch(S \cap Q)_{a':b'+1} \nsubseteq S_i$.
\begin{figure}
    \centering
    \includegraphics{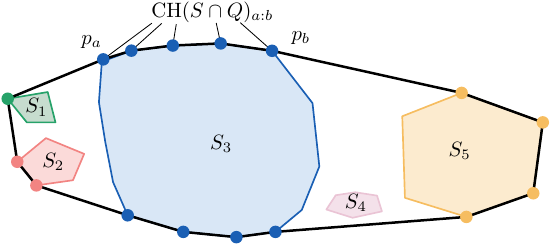}
    \caption{The convex hull consists of canonical sections, such as $\ch(S_1)_{a:b}$.}\label{fig:convex-hull}
\end{figure}

\autoref{fig:convex-hull} illustrates an example, where $\ch(S \cap Q)_{a:b} \subseteq \ch(S_3)$.
Observe that the set $S_3$ has two canonical sections, while the other sets have one, or none in case of $S_4$.
We can bound the total number of canonical sections.
\begin{restatable}{lemma}{canonicalSectionBound}\label{lemma:ch-bound}
    The convex hull of $S \cap Q$ consists of $O(m)$ canonical sections.
\end{restatable}
\begin{proof}
    The convex hulls of the sets $S_1, \dots S_m$ can be considered disjoint convex polygons $P_1, \dots, P_m$.
    Traverse the convex hull of $P_1 \cup \dots \cup P_m$ once in clockwise order.
    Consider the resulting sequence of visited polygons, \eg $P_1, P_2, \dots P_1, P_6, P_8 \dots P_1$.
    Pick any two polygons $P_i, P_j$.
    The alternating subsequence $\dots P_i, \dots P_j, \dots
    P_i, \dots P_j$ would imply that $P_i$ and $P_j$ are not convex or not disjoint.
    Hence, the traversal cannot contain such subsequences and is a Davenport-Schinzel sequence of order 2 with $m$ distinct values.
    The sequence has maximum length $2m - 1$~\cite{sharirDavenportSchinzelSequencesTheir1995}, thus we have $O(m)$ canonical sections.
\end{proof}

To derive canonical sections we compute tangents connecting sections of $\ch(S \cap Q)$ in time $O(m^2 \log n)$, \ie time $O(\log^3 n)$ with $m = O(\log n)$ as shown in~\cite{brassRangeAggregateQueriesGeometric2013}.

During preprocessing, for every secondary range tree we compute the
convex hulls of canonical sets bottom up, by merging the convex hulls
of children in linear time. 
Thus, we require $O(n \log^2 n)$ preprocessing time storage for the complete tree.

\subsection{Farthest-point Voronoi diagrams}\label{subsec:fpvd}

A well-known data structure that has close connections to the smallest
enclosing disk is the \emph{farthest-point Voronoi diagram}
~\cite{debergComputationalGeometryAlgorithms2008}.  The
farthest-point Voronoi diagram assigns the farthest point of
$S$ to every point in the plane, \ie the point $p \in S$ such that
$d(q,p) = \max_{p' \in S} d(q, p')$, as illustrated by
\autoref{fig:fpv-graph}.  With $\fpvdmth(S)$ we refer to the
farthest-point Voronoi diagram of the set of points $S$.  The diagram
$\fpvdmth(S)$ is a subdivision of the plane into $O(|S|)$ cells, such
that the farthest point does not change in a cell.
Every point $p \in \ch(S)$ has exactly one non-empty cell $\cellmth(p, S) = \bigcap_{q \in S \setminus \{p\}} H(p, q)$ where $H(p, q) = \{x \in \mathbb R^2 \mid d(x, q) \le d(x, p)\}$ is the bisector-induced closed halfspace in which $p$ is at least as far.
Points in $S \setminus \ch(S)$ have no cell.
In the remainder of the paper, we consider only points $p \in \ch(S)$ that have a cell.

\begin{figure}
    \centering
    \includegraphics*[page=7]{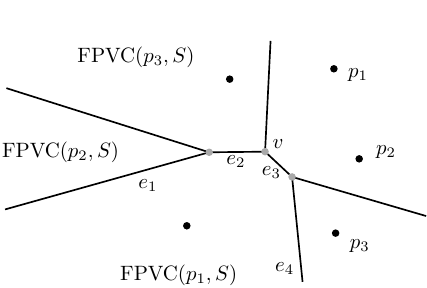}
    \caption{A farthest-point Voronoi diagram consisting of vertices, such as $v$ and edges such as $e_1, \dots, e_4$. The
      edges $e_1$ and $e_4$ are half-infinite edges.
    }\label{fig:fpv-graph}
\end{figure}
A farthest-point Voronoi diagram can be interpreted as a tree of
linear complexity~\cite{debergComputationalGeometryAlgorithms2008} as
illustrated by \autoref{fig:fpv-graph}.
A vertex of $\fpvdmth(S)$ is a point $v \in \mathbb R^2$, such that there are three points
$p_1, p_2, p_3 \in S$ with $\{v\} = \cellmth(p_1, S) \cap \cellmth(p_2, S) \cap \cellmth(p_3, S)$.
We refer to $p_1, p_2, p_3$ as the \emph{defining} points of $v$.
An edge $e$ of $\fpvdmth(S)$ is a
(half-infinite) line segment, such that two points $p_1, p_2 \in S$
exist with $e = \cellmth(p_1, S) \cap \cellmth(p_2, S)$ and the
intersection is non-degenerate (\ie contains more than one point).  We
call $p_1$ and $p_2$ the \emph{defining} points of $e$.  Let
$e_1, \dots, e_m$ be the edges that bound $\cellmth(p, S)$ in
clockwise order. Let $p_i$ be the \emph{defining} point of $e_i$;
i.e. the point that together with $p$ defines $e_i$. We then use the
following key observations:

\begin{restatable}{observation}{obsFpvRays}\label{obs:fpvrays}
  Let $r$ be a ray starting at any point $s \in \cellmth(p,
  S)$ and pointing opposite to $p \in S$. The ray
  $r$ is fully contained in $\cellmth(p, S)$; i.e. $r \subseteq
  \cellmth(p, S)$.
\end{restatable}
\begin{proof}
    By definition $\cellmth(p, S)$ is the intersection of halfspaces induced by bisectors.
    Take any of these bisectors and focus on the induced halfspaces.
    Then the halfspace $H$ where $p$ is further is the halfspace that does not contain $p$.
    By assumption $s \in H$ and hence the bisector separates $p$ and $s$.
    Since $r$ points opposite to $p$, it cannot cross the bisector and therefore remains in $H$.
\end{proof}

\autoref{obs:fpvrays} also implies that each cell is unbounded
and hence also proves that the graph $\fpvdmth(S)$ does not contain
circles~\cite{debergComputationalGeometryAlgorithms2008,aurenhammerVoronoiDiagramsDelaunay2013}. It
then follows that $e_1$ and $e_m$ are half-infinite edges (halflines),
and all other edges $e_i$ are bounded line segments. Furthermore:

\begin{figure}
    \begin{minipage}[t]{0.45\linewidth}
        \centering
        \includegraphics*[]{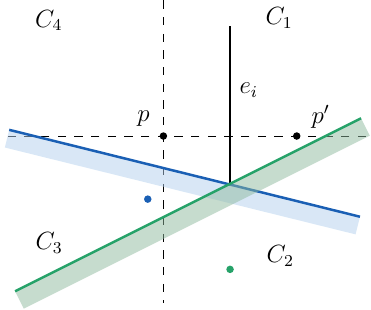}
        \caption{Illustrating the proof of \autoref{obs:neighbours}: the edge $e_i$ defined by $p$ and the clockwise neighbor $p'$ comes first in clockwise order.}\label{fig:fpv-neighbors}
    \end{minipage}
    \hfill
    \begin{minipage}[t]{0.45\linewidth}
        \centering
        \includegraphics*[page=6, trim = 35 0 0 0]{fpv-example.pdf}
        \caption{Illustrating the contradiction in the proof of \autoref{obs:order-of-edges}: the order of points $p_j, p_i$ contradicts the order of edges $e_i, e_j$.}    
        \label{fig:clockwiser-order}
    \end{minipage}
\end{figure}
\begin{restatable}{observation}{obsNeighbours}\label{obs:neighbours}
  The defining point $p_1$ for the first edge $e_1$ of any cell
  $\cellmth(p, S)$ is the clockwise neighbor of $p$ on
  $\ch(S)$. Similarly, the defining point $p_m$ for the last edge
  $e_m$ is the counter-clockwise neighbor of $p$ on $\ch(S)$.
\end{restatable}
\begin{proof}
    Let $e_1, \dots, e_m$ be the edges bounding $\cellmth(p, S)$ in clockwise order.
    We show that the clockwise neighbor $p'$ of $p$ is a defining point of $e_1$.
    Then by symmetry, it follows the counter-clockwise neighbor of $p$ defines $e_m$.
    
    W.l.o.g\@. we assume that $p$ is the origin and that $p'$ lies to the right of $p$ on the $x$-axis.
    For contradiction assume that the edge of $p$ and $p'$ is $e_i$ with $i> 1$.
    Then the edge $e_1$ precedes $e_i$ in the clockwise order.
    Let $p''$ be the point that defines $e_1$.
    The bisector of $p$ and $p'$ is parallel to the y-axis, hence $e_i$ is as well.
    Consider the four quadrants $C_1, C_2, C_3, C_4$, where $C_1$ is the region of positive $x$ and $y$ coordinates and the other are in clockwise order, as shown in \autoref{fig:fpv-neighbors}.
    The point $p''$ cannot lie in $C_4$, because then $p$ would not be a vertex of the convex hull.
    If $p'' \in C_2 \cup C_3$ then the intersection of the bisectors of $p'$ and $p$ and that of $p'$ and $p''$ lies below the $x$ axis.
    The region where $p$ is farther than $p''$ is the halfspace below this intersection point.
    But since $e_i$ lies above the intersection point, and hence in the region where $p$ is closer, this would imply that $e_1$ is no edge of $\cellmth(p, S)$ at all, a contradiction.
    Hence, $p'' \in C_1$.
    But this implies that $p''$ is the clockwise neighbor of $p$ or that $p, p'$ and $p''$ are collinear.
    The first statement contradicts our assumption while the former implies that $p'$ or $p''$ is no vertex of the convex hull and thus has no cell, a contradiction.
    Thus, an edge preceding $e_i$ does not exist.
\end{proof}
\begin{restatable}{observation}{obsOrderOfEdges}\label{obs:order-of-edges}
  Let $e_1, \dots, e_m$ be the edges of $\cellmth(p,S)$ in clockwise
  order, then the corresponding defining points $p_1, \dots, p_m$ are
  also in clockwise order on the convex hull.
\end{restatable}
\begin{proof}
    Let $\pi : \{1, \dots, m\} \rightarrow \{1, \dots, m\}$ be a permutation such that $\pi(i)$ gives the index of $p_i$ in the clockwise ordering of $p_1, \dots, p_m$, \ie the sequence $p_{\pi^{-1}(1)}, \dots, p_{\pi^{-1}(m)}$ is ordered clockwise.
    We show $\pi(i) = i$ for all $i = 1, \dots, m$.
    By \autoref{obs:neighbours} we already know $\pi(1) = 1$ and $\pi(m) = m$.
    For contradiction assume that there are $i,j$ with $\pi(i) > \pi(j)$ while $1 \le i < j \le m$.

    Let $e_1', \dots, e_{m'}'$ be the edges that bound $\cellmth(p_i, S)$ in clockwise order.
    Consider that $e_i$ is also an edge of $\cellmth(p_i, S)$, hence there is an edge $e_{i'}' = e_i$.
    Now consider rays $r$ and $r'$ that point opposite to $p$ and $p_i$ respectively and emanate from a point on $e_i$ as shown in \autoref{fig:clockwiser-order}.
    By \autoref{obs:fpvrays} they split the plane into two regions $R_1, R_2$.
    The edges $e_1, \dots, e_i$ and $e_{i}', \dots, e_{m'}'$ lie in one region, say $R_1$ while $e_i, \dots, e_j, \dots e_m$ and $e_{1}', \dots, e_{i'}'$ lie in the other region $R_2$.
    Since, $e_{m'}' \subseteq R_1$ the cell of the counter-clockwise neighbor of $p_i$ lies entirely in $R_1$ too.
    By transitivy this property extends to subsequent counter-clockwise neighbors, up to $p$.
    Since $\pi(i) > \pi(j)$ by assumption, one such neighbor is $p_j$.
    Hence, $\cellmth(p_j, S)$ lies in $R_1$ too, but this is a contradiction to $e_{j} \subseteq R_2$.
    We conclude that $\pi(i) < \pi(j)$ holds for all $i,j$ with $1 \le i < j \le m$ and thus the statement follows.
\end{proof}

\begin{figure}
    \centering
    \includegraphics*[page=4]{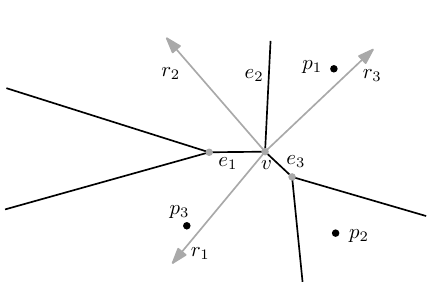}
    \caption{The defining points of a vertex admit rays that
        divide the plane and $\fpvdmth(S)$ in three parts.
    }\label{fig:fpv-rays}
\end{figure}

\begin{restatable}{observation}{fpvVertexRays}\label{obs:fpv-vertex-rays}
  Let $p_1,p_2,p_3$ be three defining points of a vertex $v$, and let
  $r_i$ be the ray starting at $v$ pointing opposite to $p_i$. The three
  rays $r_1,r_2,r_3$ divide the plane and $\fpvdmth(S)$ in three
  disjoint parts. See \autoref{fig:fpv-rays}.
\end{restatable}
\begin{proof}
    Consider the edges $e_1, e_2, e_3$ that are incident to $v$, as shown in \autoref{fig:fpv-rays}.
    We first consider two rays.
    Each edge $e_i$ is the area where the two defining points $p_a, p_b$ have the same distance, \ie their cells meet.
    Consider the rays $r_a, r_b$ that emanate from $v$ and point opposite to $p_a, p_b$ respectively.
    They divide the plane into two regions $R_1, R_2$.
    Based on \autoref{obs:fpvrays} we can show that the rays do not intersect any edge of $\fpvdmth(S)$, except $e_1, e_2, e_3$ on $v$ since they emanate on $\{v\} = e_1 \cap e_2 \cap e_3$.
    One region, say $R_1$, contains $e_i$ and the subtree reachable via $e_i$ from $v$ while $R_2$ contains the other parts of $\fpvdmth(S)$.
    Thus, when considering all three rays $r_1, r_2, r_3$ the plane is split in three regions containing one edge of $e_1, e_2, e_3$ respectively.
\end{proof}

The last observation regards the subgraph induced by a section
$\ch(S_i)_{a:b}$.
Let $G$ be the subgraph consisting of all edges and
vertices that are incident to a cell of a point of $\ch(S_i)_{a:b}$.
This subgraph has just one component, \ie is a subtree:
The cells of neighboring points of $\ch(S_i)_{a:b}$ share a half-infinite edge,
hence lie in the same component of $G$.
Since $G$ contains the boundary of every cell, all cells are connected transitively.

\begin{observation}\label{obs:section-subtree}
    For any set of points $S$ and section of the convex hull $\ch(S)_{a:b}$, the edges and vertices of \,$\fpvdmth(S)$ that are adjacent to cells of points of \,$\ch(S)_{a:b}$ form a subtree.
\end{observation}

\subparagraph{Farthest-point Voronoi diagrams and smallest enclosing
  disks.}
A well known fact regards the points defining $\sedmth(S)$.
\begin{fact}[restate=factSEDDefiningPoints,name=]
The smallest enclosing disk is defined by two or three points lying on its boundary, forming either an antipodal pair or a non-obtuse triangle.
\end{fact}
    For completeness, we provide a proof in \autoref{app:omitted_proofs}.

Since the defining points lie on the boundary, they are the farthest points from the center $z$ of $\sedmth(S)$.
If there are three defining points $p_a, p_b$ and $p_c$, then $z$ is a vertex of $\fpvdmth(S)$ defined by $p_a,p_b$ and $p_c$.
In case of an antipodal pair $p_a$ and $p_b$, the center $z$ lies on the edge defined by $p_a$ and $p_b$ and corresponds to their midpoint $\midmth(p_a,p_b)$.
\autoref{fig:center-in-fpvd} presents two examples.
Hence, the farthest-point Voronoi diagram can be interpreted as a candidate set of linear size.
In fact the first subquadratic algorithm~\cite{shamosClosestpointProblems1975} for smallest enclosing disk utilizes this property to determine $\sedmth(S)$ in time $O(n \log n)$.
\newcommand{\bigfpvexone}[1]{\includegraphics*[page=#1, trim = 0 0 5 0]{bigfpv.pdf}}
\newcommand{\bigfpvextwo}[1]{\includegraphics*[page=#1, trim = 20 80 20 75]{bigfpv.pdf}}
\begin{figure}
    \centering
    \bigfpvexone{4}
    \hfill
    \bigfpvextwo{2}
    \caption{
    The center (blue) of $\sedmth(S)$ lies on an edge of $\fpvdmth(S)$ (left) or is a vertex (right).
    }\label{fig:center-in-fpvd}
\end{figure}

For every node $v$ of a secondary range tree we compute $\fpvdmth(P(v))$ in time $O(|P(v)|)$ based on the convex hull \cite{aggarwalLineartimeAlgorithmComputing1989}.
The bound of $O(n \log^2 n)$ preprocessing time and storage for the complete tree follows with a straightforward analysis.
Based on~\cite{khareImprovedBoundsSmallest2014,brassRangeAggregateQueriesGeometric2013} we then have:
\begin{lemma}\label{lemma:preprocessing}
  Using $O(n\log^2 n)$ time and space we can build a data structure that given a query $Q$ in time $O(\log^2 n)$ yields $m = O(\log n)$ canonical sets representing $\ch(S \cap Q)$ including their convex hulls and farthest-point Voronoi diagrams. In time $O(m^2 \log n)$ we can derive $O(m)$ canonical sections of $\ch(S \cap Q)$.
\end{lemma}

\section{Answering a smallest enclosing disk query}
\label{sec:decomp-search}

In this section, we show how we can efficiently answer smallest
enclosing disk queries using farthest-point Voronoi diagrams. In
\autoref{subsec:decomp-single-set} we show that for any edge of
$\fpvdmth(S)$ we can decide which subtree contains the center of the
smallest enclosing disk.  This leads to a $O(\log n)$ procedure to
compute $\sedmth(S)$. In the context of range-aggregate queries this
search procedure is not immediately applicable, since we do not have
access to the diagram $\fpvdmth(S \cap Q)$. Instead, we can get the
points on $\ch(S \cap Q)$ as $O(\log n)$ canonical sections each of
which defines a subtree of $\fpvdmth(S \cap Q)$. In
\autoref{subsec:decomp-multi-sets} we develop a search procedure to
identify the defining points contained in a given canonical section
$\ch(S_i)_{a:b}$.  This procedure requires accessing edges of
$\fpvdmth(S \cap Q)$. We describe how we can implement these
operations efficiently in \autoref{subsec:find-sep-edge}. This then
results in a $O(\log^4 n)$ time algorithm for querying
$\sedmth(S \cap Q)$, as we argue in
\autoref{subsec:puttingthingstogether}.

\subsection{Determining the smallest enclosing disk for one
  set}\label{subsec:decomp-single-set}

Let $S$ be a set of points.
An edge or vertex of $\fpvdmth(S)$ contains the center $z$ of $\sedmth(S)$.
Since $\fpvdmth(S)$ is a tree, from any point on an edge or vertex of $\fpvdmth(S)$ there is unique path to $z$.
Given an edge $e$ of $\fpvdmth(S)$, we can decide which of the two
subtrees connected by $e$ contains $z$ using the following lemma.

\begin{figure}
    \centering
    \includegraphics*[page=2,trim=0 1 0 33]{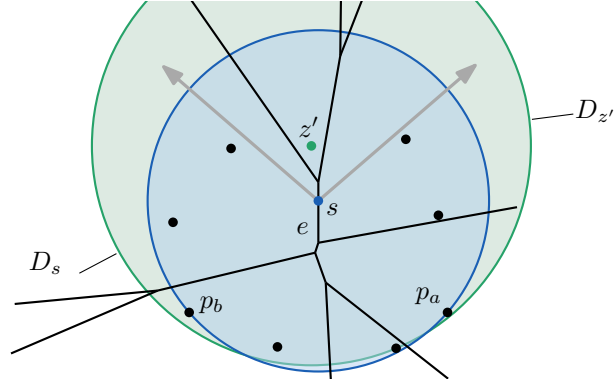}
    \caption{Any enclosing disk centered above the gray rays, such as $D_{z'}$ is larger than $D_s$.}\label{fig:fpv-observ}
\end{figure}

\begin{lemma}\label{lemma:optim-edge}
    Let $e$ be an edge of $\fpvdmth(S)$ and $p_a, p_b$ be the defining points of $e$, then for any $s \in e$ with $s \ne \midmth(p_a, p_b)$ the subtree of $\fpvdmth(S)$ reachable over $e$ by starting on $s$ and moving away from $\midmth(p_a, p_b)$ does not contain the center of $\sedmth(S)$. 
  \end{lemma}
\begin{proof}
    Let $z$ be the center of $\sedmth(S)$.
    Consider the rays $r_a$ and $r_b$ emanating from $s$ in the directions opposite to $p_a$ and $p_b$, respectively as shown in \autoref{fig:fpv-observ}.
    The rays induce two regions in the plane.
    Let $R$ be the region not containing $\midmth(p_a, p_b)$.
    Note that any point $z' \in R$ is at least as far to $p_a$ or $p_b$ as $s$, \ie $\max\{d(z', p_a), d(z', p_b)\} \ge d(s, p_a) = d(s, p_b)$.
    Hence, the enclosing disk $D_{z'}$ centered on $z'$ is at least as large as the enclosing disk $D_s$, centered at $s$.
    In case $s \ne z$, \ie $D_s$ is not optimal, we also have $z' \ne z$ by transitivity and conclude $z \notin R$.
    In case $s = z$, observe that $p_a, p_b$ are not antipodal since $s \ne \midmth(p_a, p_b)$ by assumption. Thus $s$ corresponds to an incident vertex of $e$, in particular to the vertex closer to $\midmth(p_a, p_b)$, since the other vertex admits the larger disk.
    In either case the subtree reachable by moving away from $\midmth(p_a, p_b)$ does not contain $z$ and thus the statement follows.
\end{proof}
To \emph{analyze an edge}, \ie to decide which subtree contains the center for an edge, it suffices to know the defining points and a single point on the edge.
\begin{lemma}\label{corr:optim-edge-decision}
    Given $(s, p_a, p_b)$, where $p_a, p_b \in S$ define an edge $e$ of $\fpvdmth(S)$ such that $s \in e$ we can decide in time $O(1)$ that $s$ is the center of $\sedmth(S)$ or which of the sections $\ch(S)_{a+1:b-1}, \ch(S)_{b+1:a-1}$ contains no defining points of $\sedmth(S)$.
\end{lemma}
\begin{proof}
    We check $s = \midmth(p_a, p_b)$, if this holds then $s$ is the center of $\sedmth(S)$.
    By \autoref{obs:section-subtree} the sections $\ch(S)_{a+1:b-1},
    \ch(S)_{b+1:a-1}$ induce subtrees of $\fpvdmth(S)$.
    They are connected by $e$. By comparing the relative positions of $s$ and $\midmth(p_a, p_b)$ by \autoref{lemma:optim-edge} we infer which section does not contain defining points of $\sedmth(S)$.
\end{proof}
This allows us to determine $\sedmth(S)$ based on simple comparisons of midpoints to edges.
\autoref{fig:arrow-guidance} illustrates how every edge guides towards the center of $\sedmth(S)$.
We could develop a procedure that starts at an arbitrary vertex of 
$\fpvdmth(S)$ and traverses the diagram until it encounters the edge 
or vertex defined by the defining points of $\sedmth(S)$ in time $O(n)$.
However, we aim for a more efficient procedure.

\begin{figure}
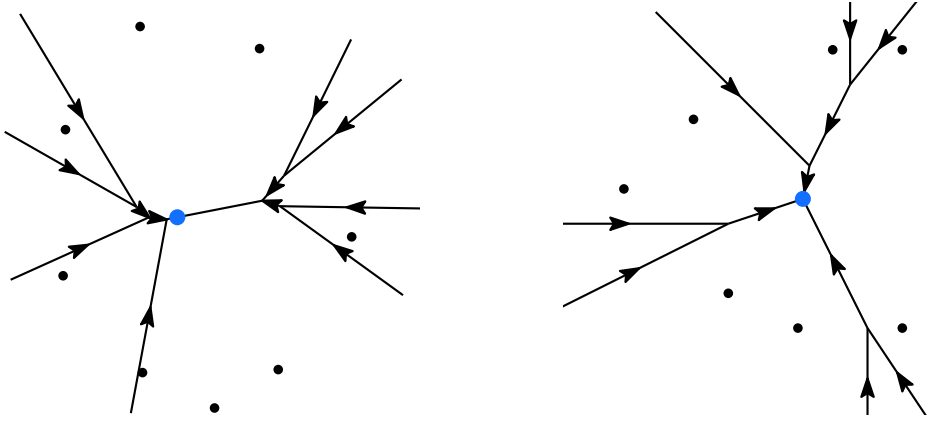

    \centering
    \bigfpvexone{3}
    \hfill
    \bigfpvextwo{1}
    \caption{Farthest-point Voronoi diagrams, where $\sedmth(S)$ is defined by two points (left) or three points (right).
    The edges are annotated according to \autoref{lemma:optim-edge}.
    }
    \label{fig:arrow-guidance}
\end{figure}

To achieve this, we require a search data structure that allows us 
to discard entire subtrees of $\fpvdmth(S)$ efficiently.
One suitable data structure is the centroid decomposition,
which we introduce next.

Let $T$ be a tree of size $n$, then a \emph{centroid} of $T$ is a vertex $v$ of $T$, such that removing $v$ splits $T$ into subtrees having a maximum size of $\frac n 2$ each.
The existence of centroids was shown in 1869 by Jordan~\cite{jordanAssemblagesLignes1869}.
We obtain the \emph{centroid decomposition} $T_C$ of $T$ by removing a centroid $v$ from $T$ and selecting $v$ as the root of the tree $T_C$. 
By removing $v$ from $T$, we created disconnected subtrees of $T$.
We recursively obtain their centroid decompositions and make their roots the children of $v$ in $T_C$.
See \autoref{subfig:decomp} for an illustration.
The resulting tree $T_C$ has height $O(\log n)$. 
To construct a centroid decomposition, there exists a folklore 
$O(n \log n)$ algorithm and a linear-time algorithm~\cite{dellagiustinaNewLinearTimeAlgorithm2019}.
During preprocessing, after computing a farthest-point Voronoi diagram in linear time, we also compute its centroid decomposition in linear time.
Thus, the bound of \autoref{lemma:preprocessing} holds.

To efficiently find an edge or vertex $x$ of a farthest-point Voronoi Diagrams $\fpvdmth(S)$, we traverse its centroid decomposition $T_C$.
To ``navigate'' the decomposition, at a node $v$ of $T_C$ we need to decide which subtree contains $x$.
Given a node $v$ of $T_C$, which corresponds to a vertex of $\fpvdmth(S)$, an oracle $\mathcal O_x(v)$ selects the edge incident to $v$ in $\fpvdmth(S)$ that lies on the path toward $x$.
If $v = x$, the oracle reports success.
Using the oracle we can efficiently identify $x$.

\begin{figure}
  \includegraphics{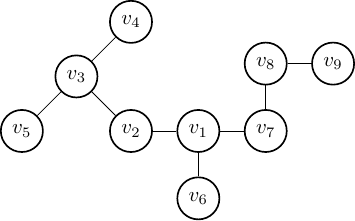}
  \hfill
  \includegraphics{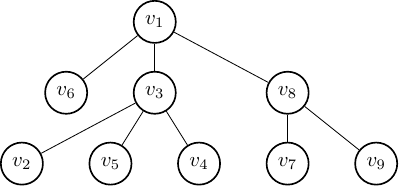}
  \caption{A tree $T$ with centroid $v_1$ (left), and its centroid decomposition $T_C$ (right).}
  \label{subfig:decomp}
\end{figure}

\begin{restatable}{lemma}{lemmaDecompSearch}\label{lemma:decompsearch}
    Given an oracle $\mathcal O_x(v)$ and a centroid decomposition of $\fpvdmth(S)$, we can find the vertex or edge $x$ of $\fpvdmth(S)$ with $O(\log n)$ invocations of $\mathcal O_x(v)$.
\end{restatable}
\begin{proof}
    Let $T_C$ be a centroid decomposition of $\fpvdmth(S)$.
    We traverse $T_C$ from its root towards a leaf.
    When visiting a node $v$, we invoke $\mathcal O_x(v)$, which reports success or yields an edge $e$ incident to $v$ in $\fpvdmth(S)$.
    Unless $\mathcal O_x(v)$ reported success or $e$ is a half-infinite edge, in which case $e = x$, we need to identify which child $w_1, \dots, w_r$ of $v$ in $T_C$ corresponds to the subtree indicated by $e$.
    Note that $r=O(1)$ for graphs of bounded degree, such as farthest-point Voronoi diagrams in general position.

    Let $p_a, p_b, p_c$ be the defining points of $v$ in clockwise order and $u$ be the other vertex of $\fpvdmth(S)$ incident to $e$.
    W.l.o.g.\@ we assume that $p_a$ and $p_b$ define $e$, thus $e$ leads to the subtree corresponding to section $\ch(S)_{a:b}$.
    For every child $w_i$ we check if the points $q_a, q_b, q_c \in S$ defining vertex $w_i$ in $\fpvdmth(S)$ lie in section $\ch(S)_{a:b}$.
    If such a child $w_i$ exists, we continue the traversal by choosing $w_i$.
    If there is no such child, $u$ is a centroid we removed during the construction of an intermediate centroid decomposition $T_C'$ and visited earlier.
    This is illustrated by \autoref{subfig:decomp}, if $\mathcal O_x(v_2) = \{v_1, v_2\}$ then there is no subtree containing $v_1$.
    The corresponding call $\mathcal O_x(u)$ must have yielded $e$ as well, since we traversed into the subtree containing the vertex $v$ incident to $e$.
    Hence, we can conclude $e = x$.
    Since the height of the centroid decomposition is $O(\log n)$, we invoke the oracle $O(\log n)$ times.
\end{proof}

Finding the smallest enclosing disk of one set of points is now straightforward, as shown in the proof of the following theorem.
\begin{theorem}\label{theorem:singlesearch}
    Given a set of points $S$ and their farthest-point Voronoi diagram in addition to its centroid decomposition, we can find the smallest enclosing disk of \,$S$ in time $O(\log |S|)$.
\end{theorem}
\begin{proof}
    Let $x$ be the vertex of $\fpvdmth(S)$ corresponding to the center $z$ of $\sedmth(S)$, in case of three defining points, or be an edge of $\fpvdmth(S)$ containing $z$ in case of two defining points.
    We use the procedure of \autoref{lemma:decompsearch} and provide an oracle $\mathcal O_x(v)$.
    The procedure implicitly maintains a subtree $T$ of $\fpvdmth(S)$ and discards parts of $T$ according to $\mathcal O_x(v)$. At any step $T$ contains $x$.
    Let $v$ be a vertex given to $\mathcal O_x(v)$ and $p_a, p_b, p_c$ be its defining points.
    We first check if $p_a, p_b, p_c$ form a \mbox{non-obtuse} triangle, in which case we report success, since no smaller disk can cover $p_a, p_b, p_c$.

    Otherwise, we analyze the incident edges.
    The points that define $\sedmth(S)$ define an edge or vertex, and
    hence lie in exactly one of the sections $\ch(S)_{a:b},
    \ch(S)_{b:c}, \ch(S)_{c:a}$.
    
    In case the defining points lie in $\ch(S)_{a:b}$, let $e$ be the corresponding incident edge defined by $p_a$ and $p_b$, as illustrated by \autoref{fig:vertex-sketch}.
    Take the point $s = v$, which lies on $e$.
    By analyzing the relative position of $s$ relative to $\midmth(p_a, p_b)$ by \autoref{corr:optim-edge-decision} we get an indication that $\ch(S)_{b+1:a-1}$ cannot contain defining points of $\sedmth(S)$.

    In general, we analyze all incident edges as described and can conclude that two out of the three sections $\ch(S)_{a+1:b-1}, \ch(S)_{b+1:c-1}, \ch(S)_{c+1:a-1}$ cannot contain defining points of $\sedmth(S)$.
    Hence, we select the incident edge corresponding to the remaining section.

    The points that define the vertex or edge $x$ that the procedure of \autoref{lemma:decompsearch} returns are the defining points of $\sedmth(S)$.
\end{proof}

\begin{figure}
    \centering
    \includegraphics*[trim=15 55 40 65]{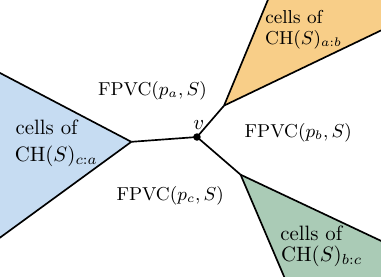}
    \caption{Conceptual sketch of a vertex $v$ of $\fpvdmth(S)$, the incident edges and the relative locations of cells.}\label{fig:vertex-sketch}
\end{figure}

\subsection{Determining the smallest enclosing disk for multiple
  sets}\label{subsec:decomp-multi-sets}

We develop a search procedure to identify the defining points contained by a given canonical section, which we invoke once for every canonical section.
We first consider the overall goal of our search procedure, \ie towards which parts of the diagrams we move.
Then we show how we guide the search procedure.

Let $G$ (``Goal'') be the set of points that define $\sedmth(S \cap Q)$.
For every canonical section $\ch(S_i)_{l:l'}$ of the convex hull we find an edge or vertex $x$ of $\fpvdmth(S_i)$ that is adjacent to the cells of $G \cap \ch(S_i)_{l:l'}$.
Whether $x$ is an edge or vertex depends on the size of $G \cap \ch(S_i)_{l:l'}$.
\[x = 
\begin{cases}
\text{an edge adjacent to } \cellmth(p_1, S_i),
& \text{if } G \cap \ch(S_i)_{l:l'} = \{p_1\}, \\
\text{the edge defined by } p_1 \text{ and } p_2,
& \text{if } G \cap \ch(S_i)_{l:l'} = \{p_1, p_2\}, \\
\text{the vertex defined by these three points},
& \text{if } |G \cap \ch(S_i)_{l:l'}| = 3.
\end{cases}
\]
In case of $|G \cap \ch(S_i)_{l:l'}| = 0$, the procedure either detects that the canonical section does not contain any defining points and aborts or returns an edge defined by the point of $\ch(S_i)_{l:l'}$ that admits the smallest enclosing disk having a point of $\ch(S_i)_{l:l'}$ on its boundary.
Once the search procedures for all $O(m)$ canonical sections is finished we gather all $O(m)$ defining points and determine the smallest enclosing disk with a linear-time algorithm.

To implement the search procedure, given a canonical section $\ch(S_i)_{l:l'}$, we apply the procedure of \autoref{lemma:decompsearch} to $\fpvdmth(S_i)$ and provide an oracle $\mathcal O_x(v)$.
Again, the procedure implicitly maintains a subtree $T$ of $\fpvdmth(S_i)$ and discards parts of $T$ according to the oracle $\mathcal O_x(v)$.
We assume $|\ch(S_i)_{l:l'}| > 2$, otherwise we can immediately return the edge defined by the point(s) of $\ch(S_i)_{l:l'}$.

Let $v$ be the vertex of $\fpvdmth(S_i)$ given to $\mathcal O_x(v)$ and $p_a, p_b, p_c \in S_i$ be the defining points of $v$ in clockwise order.
By \autoref{obs:section-subtree} the section $\ch(S_i)_{l:l'}$ induces a subtree $T$ in $\fpvdmth(S_i)$ that contains $v$.
Parts of $T$ are also present in $\fpvdmth(S \cap Q)$ as illustrated by \autoref{fig:sq-ch-subtrees}, since $\ch(S_i)_{l:l'}$ is a canonical section.
In time $O(m \log n)$ we decide if $v$ is also a vertex of $\fpvdmth(S \cap Q)$ by checking if no point is farther from $v$ than $p_a, p_b, p_c$ with queries in $\fpvdmth(S_1), \dots, \fpvdmth(S_m)$.
\begin{figure}
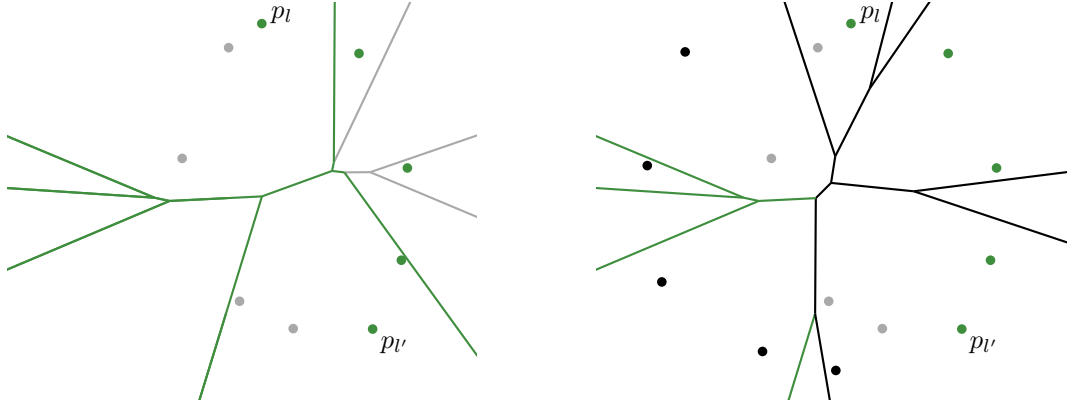

    \newcommand{\subtreefigure}[1]{\includegraphics*[trim = 45 20 55 65, page=#1]{subtrees.pdf}}
    \centering
    \subtreefigure{1}
    \hfill
    \subtreefigure{2}
    \caption{Diagram $\fpvdmth(S_i)$ (left) with a subtree induced by $\ch(S_i)_{l:l'}$ (green), parts of that tree are also present in $\fpvdmth(S \cap Q)$ (right).}\label{fig:sq-ch-subtrees}
\end{figure}

If $v$ is a vertex of $\fpvdmth(S \cap Q)$, then this implies that
$p_a, p_b, p_c$ also define incident edges of $v$ in $\fpvdmth(S \cap Q)$.
By comparing $v$ to the pairwise midpoints of $p_a, p_b, p_c$ by \autoref{corr:optim-edge-decision}, we identify which section of $\ch(S_i)_{a:b}, \ch(S_i)_{b:c}, \ch(S_i)_{c:a}$ may contain points of $G$ as in the proof of \autoref{theorem:singlesearch}.
If this section is disjoint with $\ch(S_i)_{l:l'}$ we abort the search procedure, otherwise
we resolve the call to the oracle $\mathcal O_x(v)$ with the corresponding edge.

If $v$ does not exist in $\fpvdmth(S \cap Q)$, then the incident edges defined by $p_a, p_b, p_c$ do not exist, as illustrated by \autoref{fig:diagramlayout}.
In this case the cells $\cellmth(p_a, S \cap Q), \cellmth(p_b, S \cap
Q)$ and $\cellmth(p_c, S \cap Q)$ clearly cannot meet at $v$, they are
bounded by edges defined by points not in $S_i$ (shown in red in the figure).
\newcommand{\diagramlayoutpicture}[1]{\includegraphics*[page = #1, trim = 87 80 83 57, width = \linewidth]{diagramlayout}}%
\begin{figure}
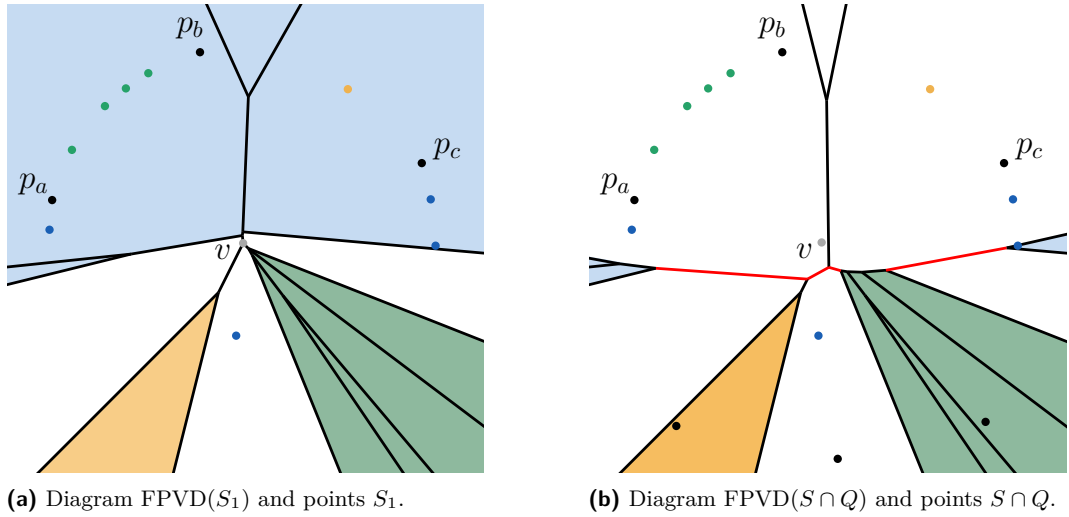

    \centering
    \subcaptionbox{Diagram $\fpvdmth(S_1)$ and points $S_1$.}
    [.45\textwidth]{\diagramlayoutpicture{1}}
    \hfill
    \subcaptionbox{Diagram $\fpvdmth(S \cap Q)$ and points $S \cap Q$. }
    [.45\textwidth]{\diagramlayoutpicture{2}}
    \caption{Comparison of cells of points of $S_1$ in $\fpvdmth(S_1)$ and $\fpvdmth(S \cap Q)$. The cells of the green points ($\ch(S_1)_{a+1:b-1}$), orange ($\ch(S_1)_{b+1:c-1}$) and blue points ($\ch(S_1)_{c+1:a-1}$) are shaded in their respective colors.
    In the latter diagram the red edges divide the cells of points of $S_1$.
    }\label{fig:diagramlayout}
\end{figure}%
Since $\ch(S \cap Q)$ contains points of other sets, the canonical section may not include all points of $S_i$.
In particular, it may not include all defining points $p_a, p_b, p_c$.
We distinguish cases, depending on how many of the defining points and their neighbors are contained in $\ch(S_i)_{l:l'}$, \ie on the number $k = \left|\left\{p_j \in \{p_a, p_b, p_c\}\mid \{p_{j-1}, p_{j}, p_{j+1}\} \subseteq \ch(S_i)_{l:l'} \right\}\right|$.

In case $k = 3$ we assume w.l.o.g.\@ that the section $\ch(S_i)_{a-1:c+1}$ is contained in the canonical section, \ie $\ch(S_i)_{a-1:c+1} \subseteq \ch(S_i)_{l:l'}$, otherwise we cyclically relabel $p_a, p_b, p_c$ accordingly.
We have the sequence
\[\ch(S_i)_{l:l'} = \{\underbrace{p_l, \dots, p_{a-1}}_{\ch(S_i)_{l:a-1}}, p_a, \underbrace{p_{a+1}, \dots, p_{b-1}}_{\ch(S_i)_{a+1:b-1}}, p_b, \underbrace{p_{b+1}, \dots, p_{c-1}}_{\ch(S_i)_{b+1:c-1}}, p_c, \underbrace{p_{c+1}, \dots, p_{l'}}_{\ch(S_i)_{c+1:l'}} \}.\]
For $p_a, p_b, p_c$ define $\prevmth(\cdot)$ and $\nextmth(\cdot)$ as all points that are before or after in the sequence relative to $p_a, p_b, p_c$ respectively, \ie $\prevmth(p_b) := \ch(S_i)_{l:b-1}$ and $\nextmth(p_b) := \ch(S_i)_{b+1:l'}$. 
We will show that $\cellmth(p_a, S \cap Q)$ has at least one \emph{separating} edge $e_a$, \ie an edge defined in conjunction with a point $p_{a'} \in \ch(S \cap Q)$ not lying in $\ch(S_i)_{l:l'}$.
We define separating edges analogously for $p_b, p_c$.
In \autoref{fig:diagramlayout} separating edges were colored red.
\begin{lemma}\label{lemma:sep-edge}${}$\\
\begin{itemize}
    \item[i)] The points $p_a$, $p_b$ and $p_c$ each have at least one separating edge.
    \item[ii)] Removing a separating edge of $p_t \in \{p_a, p_b, p_c\}$ from $\fpvdmth(S
      \cap Q)$ disconnects the cells of $\nextmth(p_t)$ from those of $\prevmth(p_t)$.
\end{itemize}
\end{lemma}

\begin{proof}
We prove the lemma for $p_b$. The other cases are analogous.

\subparagraph{i)}
Consider the sequence of edges $e_1, \dots, e_r$ 
bounding $\cellmth(p_b, S \cap Q)$ in clockwise order.
Let $q_1, \dots, q_r \in S \cap Q$ be the points that define $e_1, \dots, e_r$ respectively in conjunction with $p_b$.
By \autoref{obs:order-of-edges} the sequence $q_1, \dots, q_r$ begins with points of $\nextmth(p_b)$ and ends with points of $\prevmth(p_b)$, \ie $q_1 = p_{b+1} \in \nextmth(p_b)$ and $q_r = p_{b-1} \in \prevmth(p_b)$.
Let $q_{j}$ be the first point such that $q_{j} \notin \nextmth(p_b)$.
Such a point and corresponding edge has to exist, since otherwise $v$ would also be a vertex in $\fpvdmth(S \cap Q)$.
Similarly, let $q_{j'}$ be the last point such that $q_{j'} \notin \prevmth(p_b)$.
Note that $1 < j$ and $j' < r$ since $q_1 \in \nextmth(p_b)$ and $q_r \in \prevmth(p_b)$.
Now $j, j'$ induce three parts of the sequence: A prefix $\{q_1, \dots,
q_{j-1}\} \subseteq \nextmth{p_b}$, then $q_{j}, \dots, q_{j'}$ and
finally a suffix $\{q_{j'+1}, \dots, q_r\} \subseteq \prevmth{p_b}$.

We show that $\ch(S_i)_{l:l'}$ contains none of the points $q_{j}, \dots, q_{j'}$.
Assume one of the points, say $q_{j+1}$, lies in $\nextmth(p_b)$.
Since $q_{j} \notin \nextmth(p_b)$ and $q_1 \in \nextmth(p_b)$ this would be a contradiction to \autoref{obs:order-of-edges}.
Symmetrically $q_{j+1} \in \prevmth(p_b)$ cannot hold. 
Obviously $q_{j+1} \ne p_b$ holds as well, hence $q_{j+1} \notin \ch(S_i)_{l:l'}$.
Thus, any edge $e_s$ of $e_{j}, \dots, e_{j'}$ is a separating edge of $p_b$.

\subparagraph{ii)} Let $e_{s}$ be any edge of $e_{j}, \dots, e_{j'}$.
Removing $e_{s}$ from $\fpvdmth(S \cap Q)$ creates two subtrees.
The edges $e_1$ and $e_r$ cannot lie in the same subtree, since we removed $e_{s}$ which was part of the path from $e_1$ to $e_r$.
Let $T_1$ denote the subtree containing $e_1$ and $T_r$ denote the subtree containing $e_r$.
The section $\nextmth(p_b)$ by \autoref{obs:section-subtree} induces a subtree in $\fpvdmth(S \cap Q)$, since the point $p_{b+1} \in \nextmth(p_b)$ defines $e_1$, this subtree is contained in $T_1$, 
Similarly, the induced subtree of $\prevmth(p_b)$ lies in $T_r$.
We conclude that splitting $\fpvdmth(S \cap Q)$ on any of separating edge of $p_b$, disconnects the cells of $\prevmth(p_b)$ from those of $\nextmth(p_b)$.
\end{proof}
Let $e_a, e_b$ and $e_c$ be separating edges of $p_a, p_b$ and $p_c$,
respectively, and let $p_{a'}, p_{b'}, p_{c'}$ be the other defining
points, respectively, hence $p_{a'}, p_{b'}, p_{c'} \notin
\ch(S_i)_{l:l'}$.
By \autoref{corr:optim-edge-decision} an edge $e_j \in \{e_a, e_b, e_c\}$ with defining point $p_{j'}$ indicates $\ch(S \cap
Q)_{j+1:j'-1} \cap G = \emptyset$ or $\ch(S \cap Q)_{j'+1:j-1} \cap G
= \emptyset$ (otherwise we have found the center of $\sedmth(S \cap Q)$ and terminate the search procedure).
Since $\nextmth(p_j)
\subseteq \ch(S \cap Q)_{j+1:j'-1}$ and $\prevmth(p_j) \subseteq \ch(S \cap Q)_{j'+1:j-1}$ it follows that either $\nextmth(p_j) \cap G = \emptyset$ or $\prevmth(p_j) \cap G = \emptyset$. 

We use $\leftarrow_j$ if by indication of $e_j$ we know that $\nextmth(p_j)$ cannot contain points of $G$ and $\rightarrow_j$ in the other case.
Now consider the possible outcomes if we analyze $e_a, e_b, e_c$ this way:
\begin{itemize}
    \item Case $\leftarrow_a, \leftarrow_b, \leftarrow_c$: Then $\ch(S_i)_{a+1:l'} \cap G = \emptyset$
    \item Case $\rightarrow_a, \leftarrow_b, \leftarrow_c$: Then $\ch(S_i)_{l:a-1} \cap G = \emptyset$ and $\ch(S_i)_{b+1:l'} \cap G = \emptyset$
    \item Case $\rightarrow_a, \rightarrow_b, \leftarrow_c$: Then $\ch(S_i)_{l:b-1} \cap G = \emptyset$ and $\ch(S_i)_{c+1:l'} \cap G = \emptyset$
    \item Case $\rightarrow_a, \rightarrow_b, \rightarrow_c$: Then $\ch(S_i)_{l:c-1} \cap G = \emptyset$
\end{itemize}
Unlisted cases, such as $\leftarrow_a \rightarrow_b \rightarrow_c$ would be a contradiction to \autoref{corr:optim-edge-decision}, since $\leftarrow_a$ implies $\leftarrow_b$ and $\leftarrow_c$ due to $\nextmth(p_a) \supseteq \nextmth(p_b) \supseteq \nextmth(p_c)$.
For every case, we can conclude that only the respective ``leftover'' section $\ch(S_i)_{l:a}, \ch(S_i)_{a:b}, \ch(S_i)_{b:c}$ or $\ch(S_i)_{c:l'}$ possibly contains points of $G$.
Note that there is exactly one corresponding edge of $\fpvdmth(S_i)$ incident to $v$ for each of these sections, as illustrated by \autoref{fig:vertex-sketch}.
The edge defined by $p_a$ and $p_b$ corresponds to $\ch(S_i)_{a:b}$.
The edge defined by $p_a$ and $p_c$ corresponds to $\ch(S_i)_{l:a}$ and $\ch(S_i)_{c:l'}$ since both sections are subsets of $\ch(S_i)_{c:a}$.
Let $e$ be the edge corresponding to the section we identified.
Since $x$ is incident to a cell of a point of $G$ and hence in the subtree reachable with $e$, we select $e$ as an answer for $\mathcal O_x(v)$.

For $k \le 2$,
we must only determine and analyze $k$ separating edges.
Based on the indication given by \autoref{corr:optim-edge-decision} we identify which section $\ch(S_i)_{a:b}, \ch(S_i)_{b:c}, \ch(S_i)_{c:a}$ may contain points of $\ch(S_i)_{l:l'} \cap G$ and select the corresponding edge.
    \begin{itemize}
        \item Case $k = 2$:
        We assume w.l.o.g. $\ch(S_i)_{a-1:b+1} \subseteq \ch(S_i)_{l:l'}$, otherwise relabel $p_a, p_b, p_c$.
        Thus, we have \[\ch(S_i)_{l:l'} = \{\underbrace{p_l, \dots, p_{a-1}}_{\ch(S_i)_{l:a-1}}, p_a, \underbrace{p_{a+1}, \dots, p_{b-1}}_{\ch(S_i)_{a+1:b-1}}, p_b, \underbrace{p_{b+1}, \dots, p_{l'}}_{\ch(S_i)_{b+1:l'}} \}.\]
        By finding a separating edge of $p_a$ and $p_b$, by \autoref{lemma:sep-edge} and \autoref{corr:optim-edge-decision} we can infer which of the sections $\ch(S\cap Q)_{l:a}$, $\ch(S\cap Q)_{a:b}$, $\ch(S\cap Q)_{b:l'}$ possibly contains points of $G$.
        Note that this time $\ch(S_i)_{b:l'} \subseteq \ch(S_i)_{b:c}$ and hence the edge defined by $p_b, p_c$ corresponds to $\ch(S_i)_{b:l'}$.
        The other two edges, the one defined by $p_a, p_b$ and the one defined by $p_c, p_a$ correspond to $\ch(S_i)_{a:b}$ and $\ch(S_i)_{l:a}$ similar to before. We yield the corresponding edge of $\fpvdmth(S_i)$ incident to $v$ as an answer for $\mathcal O_x(v)$.
        \item Case $k = 1$:
        We ensure that $\{p_{b-1}, p_b, p_{b+1}\} \subseteq \ch(S_i)_{l:l'}$ holds by relabeling $p_a, p_b, p_c$ just as before.
        We have
        \[\ch(S_i)_{l:l'} = \{\underbrace{p_l, \dots, p_{b-1}}_{\ch(S_i)_{l:b-1}}, p_b, \underbrace{p_{b+1}, \dots, p_{l'}}_{\ch(S_i)_{b+1:l'}}\}.\]
        In this case we have $\ch(S_i)_{l:b-1} \subseteq \ch(S_i)_{a:b}$ and $\ch(S_i)_{b+1:l'} \subseteq \ch(S_i)_{b:c}$.
        To decide which section can contain defining points it suffices to analyze the separating edge of $p_b$.
        Accordingly, we either return the edge defined by $p_a$ and $p_b$ or the edge defined by $p_b$ and $p_c$.
        \item Case $k = 0$: Since $|\ch(S_i)_{l:l'} > 2|$, the canonical section is a subset of exactly one of the sections $\ch(S_i)_{c:a}, \ch(S_i)_{a:b}, \ch(S_i)_{b:c}$. We return the corresponding incident edge.
    \end{itemize}

We have $O(m)$ canonical sections, requiring $O(\log n)$ invocations of
$\mathcal O_x(v)$ by \autoref{lemma:decompsearch}.
Every invocation involves analyzing up to three seperating edges.
Hence, we obtain the following lemma.
In the next section, we will then show how to find a seperating edge.

\begin{lemma}
    By analyzing $O(m \log n)$ separating edges we can determine $\sedmth(S \cap Q)$.
\end{lemma}

\subsection{Finding separating edges}\label{subsec:find-sep-edge}
In this subsection we will give an $O(m \log n)$ procedure to find a point on a separating edge.

Consider the canonical sets of points $S_1, \dots, S_m$, a point
$p \in S_i$, and a vertex $v$ of $\fpvdmth(S_i)$ defined by $p$.
Additionally, we are given a canonical section $\ch(S_i)_{a:b}$ with
$p \in \ch(S_i)_{a+1:b-1}$, \ie the neighbors of $p$ lie in
$\ch(S_i)_{a:b}$.  By assumption $v$ does not exist in
$\fpvdmth(S \cap Q)$. We will describe an $O(m \log n)$ time procedure
to find a point $s$ on an edge of $\fpvdmth(S \cap Q)$ defined by $p$
and a point $p' \notin \ch(S_i)_{a:b}$, \ie a point on a separating edge.

Consider the ray $r$ emanating from $v$ and pointing opposite to $p$.
We show that $r$ intersects a separating edge of $p$.
\begin{lemma}
    The ray $r$ emanating from $v$ and pointing opposite to $p$ intersects an edge of $\cellmth(p, S \cap Q)$ defined by a point $p' \notin \ch(S_i)_{a:b}$.
\end{lemma}
\begin{proof}
    The half-infinite edges $e_1, e_2$ bounding $\cellmth(p, S \cap Q)$ are those of $p$ and its neighbors contained in $S_i$.
    Let $u,w$ be the vertices of $\fpvdmth(S \cap Q)$ that $e_1$ and $e_2$ are incident to respectively, then by convexity the line segment $uw$ lies in $\cellmth(p, S \cap Q)$.
    By \autoref{obs:fpvrays} we have $r \subseteq \cellmth(p, S_i)$, this cell is also bounded by the half-infinite edges defined by $p$ and its neighbors. Since the ray is fully contained by the induced halfspaces of $e_1$ and $e_2$ it intersects the line segment $uw \subseteq \cellmth(p, S \cap Q)$.
    The ray emanates from $v \notin \cellmth(p, S \cap Q)$.
    To intersect the line segment it must intersect an edge $e$ of $\cellmth(p, S \cap Q)$.
    The other point $p'$ that defines $e$ cannot lie in $S_i$ since the ray lies in the interior of $\cellmth(p, S_i)$.
\end{proof}
\begin{figure}
    \centering
    \includegraphics*[trim=0 0 0 0]{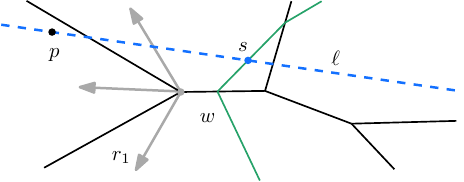}
    \caption{The point $p \in S_i$, the boundary of $\cellmth(p, S_j \cup \{p\})$ (green) and $\fpvdmth(S_j)$ (black).}\label{fig:intersec-single-diagram}
\end{figure}
Thus, there is an intersection point of $r$ and a separating edge.
To determine this intersection point, observe that $\cellmth(p, S \cap Q) = \bigcap_{j=1}^m \fpvdmth(p, S_j \cup \{p\})$ holds by definition.
We determine the intersection of $r$ and the boundary of $\fpvdmth(p, S_i \cup \{p\})$ for every $S_j \ne S_i$, then the most restrictive intersection point, \ie farthest from $v$, corresponds to a point on an edge of $\cellmth(p, S \cap Q)$.
We extend $r$ to infinity in both directions to get the line $\ell \supset r$.
\begin{lemma}
    For every set $S_j$ in time $O(\log n)$ we can find the intersection of $\ell$ and an edge $e$ of $\cellmth(p, S_j \cup \{p\})$ and its defining point $p'$.
\end{lemma}
\begin{proof}
    Let $s = \ell \cap e$ be the intersection point.
    Observe that $s \in \cellmth(p', S_j \cup \{p\}) \subseteq \cellmth(p', S_j)$, thus we identify the cell of $\fpvdmth(S_j)$ containing the not explicitly known point $s$.
    We use the search procedure of \autoref{lemma:decompsearch} to efficiently find an arbitrary edge of $\cellmth(p', S_j)$ with a technique inspired by \cite{aronovDataStructuresHalfplane2018}.
    In every step we ensure that the implicitly maintained subtree of $\fpvdmth(S_j)$ contains an edge of $\cellmth(p', S_j)$.

    Let $w$ be vertex given to the oracle $\mathcal O_x(w)$.
    By \autoref{obs:fpv-vertex-rays} the rays $r_1, r_2, r_3$ emanating from $w$ and pointing opposite to the defining points $q_1, q_2, q_3$ of $w$ divide the plane in three regions having one corresponding incident edge each as illustrated by \autoref{fig:intersec-single-diagram}.
    The region containing $s$ also contains a bounding edge $x$ of $\cellmth(p', S_j)$.
    
    To decide which region contains $s$, we may have to determine the location of $s$ relative to an intersection of $\ell$ and a ray $r_k$ of $r_1, r_2, r_3$.
    Observe that $s$ by definition splits $\ell$ into two parts, the part where $p$ is the farthest point extending along $r$, and the part where $p$ is closer than another point.
    Thus, to infer the location of $s$ relative to $\ell \cap r_k$ we must only compare the distance of $p$ and $q_k$ and choose the region we reach by moving along $r$ from $\ell \cap r_k$ exactly when $p$ is closer.

    The procedure of \autoref{lemma:decompsearch} invokes $\mathcal O_x(w)$ times $O(\log n)$, we resolve every invocation in time $O(1)$.
    Once a single edge $\fpvdmth(S_j)$ remains, the procedure of \autoref{lemma:decompsearch} returns the edge. In constant time we decide which defining point is $p'$.
\end{proof}

Thus, we can find $m$ intersection points $s_1, \dots, s_m$ and defining points $p_1, \dots, p_m$ in time $O(m \log n)$.
In time $O(m)$ we identify the point $s_j$ farthest along $r$ and return $(s_j, p, p_j)$.

\subsection{Putting things together}\label{subsec:puttingthingstogether}

To determine $\sedmth(S \cap Q)$ we derive $m = O(\log n)$ canonical sets and canonical sections of $\ch(S \cap Q)$ in time $O(\log^3 n)$ using our search data structure as in \autoref{lemma:preprocessing}.
The search procedure for every canonical set introduced in \autoref{subsec:decomp-multi-sets} has $O(\log n)$ steps.
Each step involves finding up to three separating edges. With the procedure of \autoref{subsec:find-sep-edge} we find a separating edge in time $O(m \log n)$.
Finally, we determine $\sedmth(S \cap Q)$ in time $O(m)$.
Overall,  we determine $\sedmth(S \cap Q)$ in time $O(m^2 \log^2 n)$.

\begin{theorem}
    Given sets $S_1, \dots, S_m$ with $|S_i| = O(n)$, the diagrams \,$\fpvdmth(S_1), \dots, \allowbreak{} \fpvdmth(S_m)$, their centroid decompositions and the canonical sections of the convex hull of \ $\bigcup_{i=1}^m S_i$, in time $O(m^2 \log^2 n)$ we can find $\sedmth(\bigcup_{i=1}^m S_i)$, assuming the convex hull consists of \,$O(m)$ canonical sections.
\end{theorem}

\begin{theorem}\label{theorem:decompconstantsets}
    Given a constant number of sets having a constant number of canonical sections, we can find the smallest enclosing disk in time $O(\log^2 n)$.
\end{theorem}
In the context of range-aggregate queries, the number of sets is $m = O(\log n)$ and the convex hull consists of $O(m)$ canonical sections by \autoref{lemma:ch-bound}.
This yields the following theorem.
\begin{theorem}
    Let $S$ be a set of points of size $n$ in the plane.
    With $O(n \log^2 n)$ preprocessing time and storage we can answer orthogonal range-aggregate queries for the smallest enclosing disk in time $O(\log^4 n)$.
\end{theorem}

We note that the search procedure is versatile enough to find other points on the graph $\fpvdmth(S \cap Q)$, provided we can determine which subtree contains them, given an edge.

\section{Randomized approaches}
\label{sec:Randomized_Approaches}

A well known randomized algorithm that determines $\sedmth(S)$ given a
set of $n$ points $S$ in expected linear time is Welzl's
algorithm~\cite{welzlSmallestEnclosingDisks1991}. A natural idea is to
try to generalize this algorithm to sets of disjoint convex
polygons. In
\autoref{sub:An_attempt_at_Randomized_Incremental_Construction} we
present a counter example that shows that a direct adaptation
unfortunately does not work. However, we can show that we can plug in
our implicit searching procedure from the previous section into
Eppstein's randomized dynamic programming framework, see
\autoref{sub:Using_Eppstein's_Dynamic_Programming_framework}, to
answer smallest enclosing disk queries in expected
$O(\log^{5/2}n\log\log n)$ time instead.

\subsection{An attempt at randomized incremental construction}
\label{sub:An_attempt_at_Randomized_Incremental_Construction}

Welzl's algorithm~\cite{welzlSmallestEnclosingDisks1991} builds upon
three central properties:
\begin{itemize}
    \item $\sedmth(S)$ is defined by up to three points.
    \item By testing $p \notin \sedmth(S \setminus \{p\})$, we can decide whether any point $p \in S$ is a defining point.
    \item The probability that a random point $p \in S$ is a defining point is $O(\frac{1}{n})$.
\end{itemize}
This allows for a simple procedure: Consider the points $p_1, \dots, p_n \in S$ in random order and maintain an intermediate disk $D_i$ at every step.
When considering $p_{i+1}$, in case of $p_{i+1} \in D_i$ set $D_{i+1} := D_i$, otherwise $p_{i+1}$ is a defining point and lies on the boundary of $D_{i+1}$.
Recurse to construct $D_{i+1} := \md(S', R)$ with $S' = \{p_1, \dots, p_i\}$ and $R = \{p_{i+1}\}$.
The disk $\md(S', R)$ denotes the smallest disk covering the points $S' \cup R$, while having all points of $R$ on its boundary.
Any recursive subproblem $\md(S', R)$ can be solved with the same strategy. For $|R| = 3$ the disk is uniquely defined.
The probability that $p_{i+1} \notin D_i$ holds is $O(\frac {1}{i})$, leading to $O(n)$ expected runtime when solving $\md(S, \emptyset) = \sedmth(S)$.

\SetKwFunction{SetMiniDisk}{SetMiniDisk}%
Given canonical sets $S_1, \dots, S_m$ we can test $S_i \subseteq D$ for any disk $D$ in time $O(\log n)$, by checking if $D$ contains the point that is farthest from its center.
In case $S_{i+1} \nsubseteq \sedmth(S_1 \cup \dots \cup S_{i})$ we can conclude $S_{i+1}$ contains at least one defining point.
    Thus, a natural idea is to lift Welzl's algorithm to operate on sets, instead of single points.
\newcommand{\setminidiskprocedure}[0]{
\begin{procedure}
    \caption{SetMiniDisk($\mathcal S, \mathcal R$)}
    \SetKw{KwOr}{or}
    \uIf{$|\mathcal R| = 3$ \KwOr $|\mathcal S| = 0$  }{
        \Return{$\sedmth(\bigcup_{S_i \in \mathcal R} S_i)$}\;
    }
    Choose a set $S \in \mathcal S$ uniformly at random\;
    $D \gets \SetMiniDisk(\mathcal S \setminus \{S\}, \mathcal R)$\;
    \uIf{$S \nsubseteq D$}{
        $D \gets \SetMiniDisk(\mathcal S \setminus \{S\}, \mathcal R \cup \{S\})$\;
    }
    \Return{$D$}\;
\end{procedure}%
}

    Let $\mathcal S$ and $\mathcal R$ be disjoint selections of $S_1, \dots, S_m$, \ie $\mathcal S, \mathcal R \subseteq \{S_1, \dots, S_m\}$ with $\mathcal S \cap \mathcal R = \emptyset$ and $|\mathcal{R}| \le 3$.
    The set $\mathcal{R}$ will contain sets that contain defining points of the smallest enclosing disk.
    Let $P(\mathcal S, \mathcal R)$ be the set of points contained in the inner sets, \ie $P(\mathcal S, \mathcal R) = \bigcup_{S \in \mathcal S} S\cup \bigcup_{S \in \mathcal R} S$.
    We refer to a selection of sets $\mathcal C \subseteq \mathcal R \cup \mathcal S$ as a \emph{base case}, if it contains every set of $\mathcal R$ and up to $3 - |\mathcal R|$ sets of $\mathcal S$.
    Further, if the smallest enclosing disk of the points in the sets of $\mathcal C$ covers all points, then $\mathcal C$ is a \emph{feasible} base case.
    We denote the set of feasible base cases with
    \[\fcases(\mathcal S, \mathcal R) = \left\{\mathcal C = \mathcal R \cup \mathcal S' \mid \mathcal S' \subseteq \mathcal S \land  P(\mathcal S, \mathcal R) \subseteq \sedmth(\bigcup_{S_i \in \mathcal C}S_i) \land |\mathcal C| \le 3\right\}.\]
    We say that a feasible base case $\mathcal C \in \fcases(\mathcal S, \mathcal R)$ \emph{yields} the disk $\sedmth(\bigcup_{S_i \in \mathcal C}S_i)$.

    \begin{figure}
        \centering
        \includegraphics*[page=2]{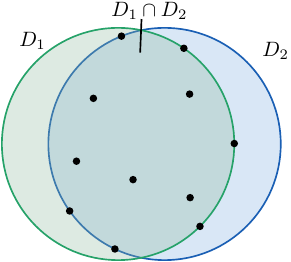}
        \caption{Two minimal disks, all points lie in $D_1 \cap D_2$. Hence, the defining points of each disk cannot define a \mbox{non-obtuse} triangle or an antipodal pair, contradicting minimality.}\label{fig:non_min_md}
    \end{figure}
    \begin{restatable}{lemma}{randSetSedUnique}\label{lemma:rand_sed_unique}
        All feasible base cases of \,$\mathcal S$ and $\mathcal R$ yield the same disk.
    \end{restatable}
    \begin{proof}
        Let $\mathcal C_1, \mathcal C_2 \in \fcases(\mathcal S, \mathcal R)$ be two feasible base cases yielding the disk $D_1$ and $D_2$ respectively.
        We first show that $D_1$ and $D_2$ have the same size, afterward we show $D_1 = D_2$.

        Assume $D_1$ and $D_2$ have differing radii, then the existence of the smaller disk implies that the larger disk was not the smallest enclosing disk of the points of the set of its base case, a contradiction.
        
        Thus assume that $D_1$ and $D_2$ have the same size.
        Observe that both disks cover all points, hence $P(\mathcal S, \mathcal R) \subseteq D_1 \cap D_2$ as shown by \autoref{fig:non_min_md}.
        Note that $\partial D_1 \cap D_2$ and $D_1 \cap \partial D_2$ are arcs of less than 180° which contain the defining points of $D_1$ and $D_2$ respectively.
        Thus, the defining points of either disk cannot form a \mbox{non-obtuse} triangle or antipodal pair, a contradiction.
    \end{proof}
    We denote the disk yielded by the feasible base cases with $\smd(\mathcal S, \mathcal R)$.
    If $\fcases(\mathcal S, \mathcal R) = \emptyset$ then $\smd(\mathcal S, \mathcal R)$ does not exist.
    For example, if $\mathcal R$ contains a set without a defining point, there may be no feasible base case, as shown in \autoref{fig:undefined_md}.
    Observe that $\sedmth(\bigcup_{S_i \in \mathcal S}S_i)$ is the special case $\smd(\mathcal S, \emptyset)$.
    \begin{figure}
        \centering
        \includegraphics*{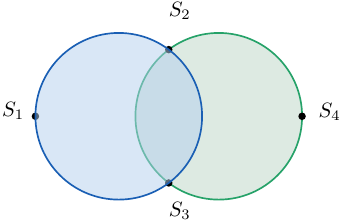}
        \caption{The disk $\smd(\{S_1, S_4\}, \{S_2, S_3\})$ is not defined. Every set of $S_1, \dots, S_4$ contains a single point. Since any base case has to include $S_2$ and $S_3$, only one of the sets $S_1$ and $S_4$ can be chosen. This allows the blue disk $\sedmth(S_1 \cup S_2 \cup S_3)$ or the green disk $\sedmth(S_2 \cup S_3 \cup S_4)$, both of which do not cover all points.}\label{fig:undefined_md}
    \end{figure}

    The following lemma, based on \cite{welzlSmallestEnclosingDisks1991,debergComputationalGeometryAlgorithms2008}, forms the central part of the algorithm.
    \begin{restatable}{lemma}{randSetSed}\label{lemma:rand_sed}${}$\\
        Let $S \in \mathcal S$. Provided $\smd(\mathcal S, \mathcal R)$ and $\smd(\mathcal S \setminus \{S\}, \mathcal R)$ exist
        \begin{enumerate}
            \item[i)] If $S \subseteq \smd(\mathcal S \setminus \{S\}, \mathcal R)$, \, $\smd(\mathcal S, \mathcal R) = \smd(\mathcal S \setminus \{S\}, \mathcal R)$.
            \item[ii)] If $S \nsubseteq \smd(\mathcal S \setminus \{S\}, \mathcal R)$, \, $\smd(\mathcal S, \mathcal R) = \smd(\mathcal S \setminus \{S\}, \mathcal R \cup \{S\})$.
        \end{enumerate}
    \end{restatable}
    \begin{proof}
        ~
        \paragraph*{i)}
        Let $D_1 = \smd(\mathcal S, \mathcal R)$ and $D_2 = \smd(\mathcal S \setminus \{S\}, \mathcal R)$.
        Since $D_2$ covers $P(\mathcal S, \mathcal R)$, any feasible base case that yields $D_2$ is also a feasible base case of $D_1$, \ie $\fcases(\mathcal S \setminus \{S\}, \mathcal R) \subseteq \fcases(\mathcal S, \mathcal R)$.
        By \autoref{lemma:rand_sed_unique} all feasible base cases yield the same disk, since $\fcases(\mathcal S \setminus \{S\}, \mathcal R)$ is not empty we conclude $D_1 = D_2$.
        \paragraph*{ii)}
        Let $D_1 = \smd(\mathcal S, \mathcal R)$ and $D_2 = \smd(\mathcal S \setminus \{S\}, \mathcal R)$.
        We will show that all feasible base cases $\mathcal C \in \fcases(\mathcal S, \mathcal R)$ contain $S$, \ie $S \in \mathcal C$.
        
        Let $\mathcal C \in \fcases(\mathcal S, \mathcal R)$ be any feasible base case giving $D_1$.
        For contradiction assume that $\mathcal C$ does not contain $S$.
        Since $\mathcal C$ does not include $S$ and $P(\mathcal S \setminus \{S\}, \mathcal R) \subseteq D_1$ it also a feasible base case of $\mathcal S \setminus \{S\}, \mathcal R$ by definition.
        Hence, if $S \notin \mathcal C$ then $D_1 = D_2$ by \autoref{lemma:rand_sed_unique}, which is a contradiction since $S \subseteq D_1$ but $S \nsubseteq D_2$.
        
        We conclude that any feasible base case $\mathcal C$ that yields $D_1 = \smd(\mathcal S, \mathcal R)$ contains $S$.
        This implies $\smd(\mathcal S, \mathcal R) = \smd(\mathcal S \setminus \{S\}, \mathcal R \cup \{S\})$ since $\mathcal C$ is also a feasible base case of $\mathcal S \setminus \{S\}, \mathcal R \cup \{S\}$.
    \end{proof}

    To build an algorithm based on \autoref{lemma:rand_sed} we require the following assumption.
    \begin{assumption}\label{assump:lessermd}
        If $\smd(\mathcal S, \mathcal R)$ exists, then $\smd(\mathcal S \setminus \{S\}, \mathcal R)$ exists for every $S \in \mathcal S$.
    \end{assumption}
    For now, we suppose that \autoref{assump:lessermd} holds.
    However, we will later show that it does not hold in general.
    Lifting Welzl's algorithm to operate on sets is now straightforward with \autoref{lemma:rand_sed} as shown in \SetMiniDisk.
    \setminidiskprocedure{}

    To find the smallest enclosing disk of $S_1, \dots, S_m$ we invoke \SetMiniDisk$(\{S_1, \dots, S_m\}, \emptyset)$.
    Since $\smd(\mathcal S, \emptyset)$ exists, by \autoref{assump:lessermd} and \autoref{lemma:rand_sed} it follows that any intermediate disk considered by \SetMiniDisk also exists.

\newcommand{\smdcounterexample}[1]{%
\fbox{\includegraphics*[page=#1, width=0.31\linewidth]{ce.pdf}}%
}
\newcommand{\smdcounterexamplebody}[0]{
\centering%
    \smdcounterexample{1}
    \smdcounterexample{2}
    \smdcounterexample{3}
}
\begin{figure}
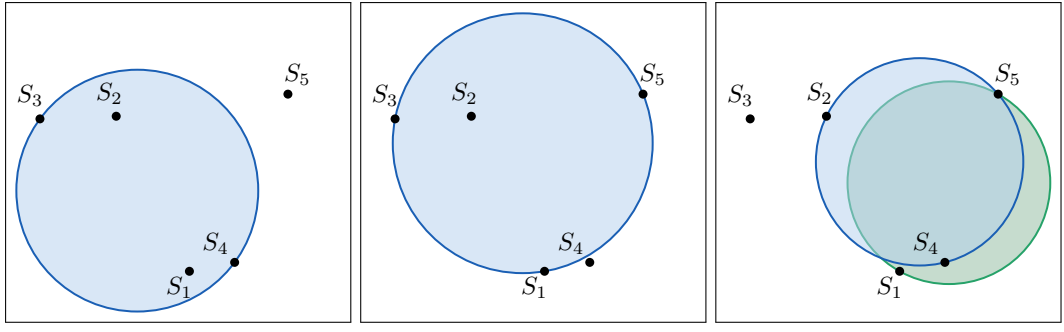

    \smdcounterexamplebody{}
    \caption{Example showing that not any intermediate disk is defined. From left to right: $\smd(\{S_1, \dots, S_4\}, \emptyset)$, $\smd(\{S_1, S_2, S_3\}, \{S_5\})$, $\smd(\{S_1, S_2\}, \{S_4, S_5\})$}\label{fig:rand-ce}
\end{figure}

    \paragraph*{Counterexample}
    We now consider why \autoref{assump:lessermd} does not hold in general.
    \autoref{fig:rand-ce} shows an example involving the sets $\mathcal S = \{S_1, \dots, S_5\}$ with $S_1 = \{p_1\}$, \dots, $S_5 = \{p_5\}$.
    We draw the sets in order $S_5, \dots, S_1$ and accordingly check $S_5 \nsubseteq \smd(\{S_1, \dots, S_4\}, \emptyset)$. This check succeeds as shown the first frame.
    Hence, by \autoref{lemma:rand_sed} we conclude \[\smd(\{S_1, \dots, S_5\}, \emptyset) =\allowbreak \smd(\{S_1, \dots, S_4\}, \{S_5\})\] and recurse into the call \SetMiniDisk$(\{S_1, \dots, S_4\}, \{S_5\})$.
    Within this call we draw the set $S_4$ and construct $\smd(\{S_1, S_2, S_3\}, \{S_5\})$ recursively.
    The second frame shows that $S_4$ is not covered by $\smd(\{S_1, S_2, S_3\}, \{S_5\})$.
    We conclude \[\smd(\{S_1, \dots, S_4\}, \{S_5\}) = \allowbreak \smd(\{S_1, S_2, S_3\}, \{S_4, S_5\})\] and recurse into the call to solve the latter subproblem.
    We draw $S_3$ and recurse into \SetMiniDisk$(\{S_1, S_2\}, \{S_4, S_5\})$, which attempts to construct the corresponding disk.
    This disk does not exist, as illustrated by the third frame.
    By definition $\smd(\{S_1, S_2\}, \{S_4, S_5\})$ is the smallest enclosing disk of $\mathcal R = \{S_4, S_5\}$ and $3- |\mathcal R| = 1$ sets of $\{S_1, S_2\}$.
    Hence, there are two candidate selections with the disks $\sedmth(S_2 \cup S_4 \cup S_5)$ (blue) and $\sedmth(S_1 \cup S_4 \cup S_5)$ (green).
    However, neither selection is a feasible base case, \ie neither covers all points of $S_1 \cup S_2 \cup S_4 \cup S_5$.
    Either $S_1$ or $S_2$ is not covered.

    Under \autoref{assump:lessermd} based on backwards analysis one may argue, that when drawing set $S_m$ of random ordered sets $S_1, \dots, S_{m}$ and having disk $D = \smd(\{S_1, \dots, S_{m-1}\}, \mathcal R)$ the probability of check $S_m \nsubseteq D$ holding and having to make an expensive recursive call is $O(\frac{1}{m})$, since $m -3$ sets are covered by $D$.
    However, in the example of \autoref{fig:rand-ce} we can force a lower bound of this probability arbitrarily close to $\frac {1}{2}$ by adding sets with points lying in the neighborhood of the point of $S_1$ and $S_2$: If the last constructed disk covers $S_1$ then it does not cover $S_2$ and its neighborhood and vice versa.

    This also demonstrates why Welzl's algorithm defines $\smd(S, R)$ to be the minimal disk covering all points and having all points of $R$ on its boundary.
    In the example, the corresponding disk constructed by Welzl's algorithm is shown in \autoref{fig:rand-ce-welzl}.
    Since $p_4$ and $p_5$ have to lie on the boundary, the disk $\md(\{p_1, p_2\}, \{p_4, p_5\})$ is not the smallest enclosing disk of $p_1, p_2, p_4, p_5$.
    \begin{figure}
        \centering
        \includegraphics*[page=4]{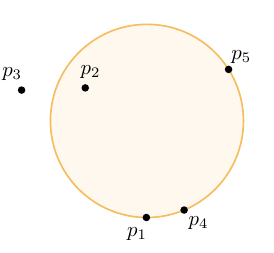}
        \caption{Disk $\md(\{p_1, p_2\}, \{p_4, p_5\})$ that would be constructed by Welzl's algorithm applied to the example of \autoref{fig:rand-ce}.}\label{fig:rand-ce-welzl}
    \end{figure}

    A natural modification is to align the definition more closely with the definition of $\md(S, R)$ in Welzl's algorithm by requiring specific points to lie on the boundary of the disk.
    Hence, define $\smd(\mathcal S, R)$, with $\mathcal S \subseteq \{S_1, \dots, S_m\}$ and $R \subseteq S_1 \cup \dots \cup S_m$ as the smallest disk that has all points of $R$ on its boundary and covers all points $P(\mathcal S)$, where $P(\mathcal S) = \bigcup_{S \in \mathcal S} S$.
    A consequence of this definition is that with $|R| = 3$ or $\mathcal S = \emptyset$ the points $R$ define a disk we can determine in constant time.
    In this approach any intermediate disk is defined.
    \begin{lemma}
        If $\smd(\mathcal S, R)$ exists, then $\smd(\mathcal S \setminus \{S\}, R)$ exists for every $S \in \mathcal S$.
    \end{lemma}
    \begin{proof}
        The disk $\smd(\mathcal S, R)$ covers $P(\mathcal S) \cup R$ and has the points of $R$ on its boundary.
        Hence, there is a non-empty set of disks $\mathcal D$ that cover $P(\mathcal S \setminus \{S\})$ and have the points $R$ on its boundary.
        The set $\mathcal D$ has minimal elements:
        If $|R| = 3$, then $\mathcal D$ has just one element.
        If $|R| = 2$, as illustrated by \autoref{fig:r2} we can reduce the size of $\smd(\mathcal S, R)$ while keeping the points $R$ on its boundary until a third point lies on the boundary, or the points of $R$ become an antipodal pair.
        In case of $|R| = 1$, let $q \in R$ be the point in $R$.
        Consider the set of disks $\mathcal D' \subseteq \mathcal D$ we obtain by shrinking the disks $\mathcal D$ until an antipodal pair or three points lie on its boundary involving $q$ respectively.
        Any disk $D \in \mathcal D \setminus \mathcal D'$ can be shrunk further, while keeping $q$ on its boundary and is therefore not minimal.
        Hence, only disks in $\mathcal D'$ can be minimal.
        Since $S$ has finite cardinality, the set $\mathcal D'$ has finite cardinality and contains a disk with minimal radius.
    \end{proof}
    The main difficulty of this approach lies in efficiently determining
    which points should be chosen.
    \begin{figure}
        \centering
        \includegraphics*{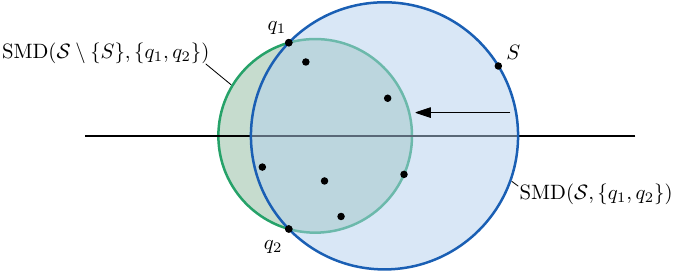}
        \caption{The center of any disk that has $q_1, q_2$ on its boundary lies on the bisector of $q_1, q_2$. The smaller disk $\smd(\mathcal S \setminus \{S\}, \{q_1, q_2\})$ can be found by shrinking $\smd(\mathcal S, \{q_1, q_2\})$, while keeping $q_1, q_2$ fixed on the boundary, until a third point lies on the boundary.}\label{fig:r2}
    \end{figure}

\subsection{Using Eppstein's dynamic programming framework}
\label{sub:Using_Eppstein's_Dynamic_Programming_framework}
Eppstein presented a framework for linear programming in the intersection of $k$ polyhedra in $\mathbb R^3$.
The dynamic programming based framework determines the optimum in expected time $O(k \log k \log n + \sqrt{k}\log k \log^3 n)$ ~\cite{eppsteinDynamicThreeDimensionalLinear1992}.
To adapt the framework, we test if a set $S_i$ contains a boundary point of $\sedmth(S \cap Q)$, by checking whether $S_i \nsubseteq \sedmth(\bigcup_{j \ne i}S_j)$ holds.
\begin{lemma}\label{lemma:boundaryPointCheck}
    If $S_i \nsubseteq \sedmth(\bigcup_{j \ne i}S_j)$ then $S_i$ contains a boundary point of $\sedmth(\bigcup_{j=1}^m S_j)$.
\end{lemma}
\begin{proof}
    Repeated element-wise application of \cite[Lemma 4.14]{debergComputationalGeometryAlgorithms2008}.
\end{proof}
The probability is $O(\frac {1}{m})$, when $S_1 \dots, S_m$ are ordered randomly.
We compute $\sedmth(\bigcup_{j \ne i}S_j)$ recursively, by leaving out single sets and computing the smallest enclosing disk for the remaining sets.
    \SetKwFunction{EMD}{DMD}%
    Based on the framework we get \autoref{alg:eppstein} (``DynamicMiniDisk''), which is implicitly a dynamic program.
    A call $\EMD(x,y,z)$ computes the smallest enclosing disk of $S_1, \dots, S_z$ and $S_{y+z}, S_{x+y+z}$.
    Hence, the call $\EMD(1,1,m-2)$ determines $\sedmth(S_1 \cup \dots \cup S_m)$.

    \begin{procedure}[]
    \caption{DMD($x,y,z$)}\label{alg:eppstein}
    \SetAlgoRefName{DMD}
    \uIf{$z = 1$}{
        \Return{$\sedmth(S_x \cup S_y \cup S_z)$}\;
    }
    \uIf{$y = 1$}{
        \lIf{$S_{x+y+z} \subseteq \EMD(y,1,z-1)$}{
            \Return{$\EMD(y,1,z-1)$}
        }
        \lIf{$S_{y+z} \subseteq \EMD(x+y,1,z-1)$}{
            \Return{$\EMD(x+1,1,z-1)$}
        }
        \lIf{$S_{z} \subseteq \EMD(x,y+1,z-1)$}{
            \Return{$\EMD(x,y+1,z-1)$}
        }
        \Return{$\sedmth(S_x \cup S_y \cup S_z)$}\;
    }
    \uElse{
        \lIf{$S_{z} \subseteq \EMD(x,y+1,z-1)$}{
            \Return{$\EMD(x,y+1,z-1)$}
        }
        \lIf{$S_{y+z} \subseteq \EMD(x+y,1,z-1)$}{
            \Return{$\EMD(x+1,1,z-1)$}
        }
        \lIf{$S_{x+y+z} \subseteq \EMD(y,1,z-1)$}{
            \Return{$\EMD(y,1,z-1)$}
        }
        \Return{$\sedmth(S_x \cup S_y \cup S_z)$}\;
    }
    \end{procedure}
    The expected number of evaluated cells is $O(m \log m)$ as shown in~\cite{eppsteinDynamicThreeDimensionalLinear1992}. Each cell involves $O(1)$ containment checks of time $O(\log n)$.
    The expected number of solved base cases, \ie $\sedmth(S_x \cup S_y \cup S_z)$ is $O(\sqrt{m}\log m)$ by~\cite{eppsteinDynamicThreeDimensionalLinear1992}.
    By \autoref{theorem:decompconstantsets} we solve a base case in time $O(\log^2 n)$.
    Thus, the expected time is $O(m \log m \log n) + O(\sqrt{m} \log m \log^2 n)$.
\begin{theorem}\label{theorem:eppsteinframework}
    Given sets $S_1, \dots, S_m$ with $|S_i| = O(n)$, the diagrams \,$\fpvdmth(S_1), \dots, \allowbreak{} \fpvdmth(S_m)$ and their centroid decompositions, we can find \,$\sedmth(\bigcup_{i=1}^m S_i)$ in expected time $O(m \log m \log n + \sqrt{m} \log m \log^2 n)$, assuming the convex hull of every triplet of sets $(S_i, S_j, S_k)$ consists of $O(1)$ canonical sections.
\end{theorem}
In the context of range-aggregate queries, the number of sets is $m = O(\log n)$ and the convex hull of every triplet has $O(1)$ canonical sections.
This leads to the following theorem.
\begin{theorem}
    Let $S$ be a set of points of size $n$ in the plane.
    With $O(n \log^2 n)$ preprocessing time and storage we can answer orthogonal range-aggregate queries for the smallest enclosing disk in expected time $O(\log^{\frac{5}{2}}n \log \log n)$.
\end{theorem}

\bibliography{references}
\clearpage
\appendix

\section{Omitted Proofs}
\label{app:omitted_proofs}

\factSEDDefiningPoints*
\begin{proof}
    We first show that a disk $D$ covering all points with an antipodal pair or points defining a non-obtuse triangle on its boundary is minimal.
    Afterward we show that a disk $D$ is not minimal otherwise.
    
    Let $S$ be a set of points.
    Let $D$ be a disk with center $z$ covering all points with an antipodal pair or points defining a non-obtuse triangle on its boundary.
    With $D(p)$ we denote the area of points that are closer to $p \in S$ than $z$, \ie $D(p) = \{q \in \mathbb R^2 \mid d(p, q) < d(p, z) \}$, which is an open disk. 

    In case of an antipodal pair $\{p_a, p_b\}$ the center $z$ is their midpoint. Hence, $D(p_a) \cap D(p_b) = \emptyset$, \ie any other center is further away from $p_a$ or $p_b$.
    In case of a non-obtuse triangle let $p_a, p_b, p_c$ be the points forming the triangle.
    Note that $z$ is the circumcenter of $\triangle p_a p_b p_c$.
    If $\triangle p_a p_b p_c$ is a right triangle, then $z$ is the midpoint of two of the points $p_a, p_b, p_c$ forming an antipodal pair, and thus $D$ is minimal.
    If $\triangle p_a p_b p_c$ is no right triangle, it must be an acute triangle, and its circumcenter $z$ lies in the interior.
    Consider the line $\ell$ that passes through $z$ and $p_a$ and focus on $D(p_a) \cap D(p_b)$ and $D(p_a) \cap D(p_c)$.
    Since $z$ lies inside the triangle, $\ell$ separates $p_b$ and $p_c$.
    The intersection $D(p_a) \cap D(p_b)$ lies on the side of $\ell$ that contains $p_b$, the intersection $D(p_a) \cap D(p_c)$ lies on the side of $\ell$ that contains $p_c$.
    Hence, $D(p_a) \cap D(p_b) \cap D(p_c)$ is empty, \ie there exists no smaller disk covering $p_a, p_b, p_c$.

    Let $D$ be the disk covering all points, while having no points, no antipodal pair or no points forming a non-obtuse triangle on its boundary.
    Then all boundary points are contained in an arc of the boundary of less than 180°. By infinitesimally shifting $D$ towards the central point of the arc and infinitesimally shrinking $D$ we constructed a smaller disk covering all points, thus $D$ is not minimal. A contradiction.
\end{proof}
\end{document}